\documentclass[11pt]{article}

\usepackage{fullpage}
\usepackage[utf8]{inputenc}
\usepackage{todonotes}
\usepackage[linesnumbered,ruled,vlined]{algorithm2e}
\usepackage{mathrsfs}
\usepackage{dsfont}
\usepackage{enumitem}
\SetKw{Break}{break}
\SetKw{AND}{and}
\SetKw{OR}{or}

\usepackage{graphics}
\usepackage[dvips]{epsfig}
\usepackage{comment}
\usepackage{xcolor}
\usepackage{amsmath}
\usepackage{amssymb}
\usepackage{amsfonts}
\usepackage{graphicx, xspace}

\newcommand{\cev}[1]{\reflectbox{$\vec{\reflectbox{\!$#1$}}$}}

\newcommand{\true}{\mbox{\it true}}
\newcommand{\false}{\mbox{\it false}}

\newtheorem{theorem}{Theorem}[section]
\newtheorem{lemma}{Lemma}[section]
\newtheorem{corollary}{Corollary}[section]
\newtheorem{claim}{Claim}[section]

\newcommand{\qed}{\hfill $\Box$ \bigbreak}
\newenvironment{proof}{\noindent {\bf Proof.}}{\qed}

\newenvironment{proofclaim}{\noindent{\bf Proof of the claim.}}{\hfill$\star$}

\usepackage{accents}
\newcommand*{\bigdot}{\accentset{\mbox{\large\bfseries .}}}

\title{{\bf Can Like Attract Like? A Study of Homonymous Gathering in Networks}}

\author{
St\'ephane Devismes\thanks{
MIS Lab., Universit\'{e} de Picardie Jules Verne, France. E-mail: stephane.devismes@u-picardie.fr}
\and
Yoann Dieudonn\'{e}\thanks{
MIS Lab., Universit\'{e} de Picardie Jules Verne, France. E-mail: yoann.dieudonne@u-picardie.fr}
\and
Arnaud Labourel\thanks{
Aix Marseille Univ, CNRS, LIS, Marseille, France. Email: arnaud.labourel@lis-lab.fr}
}
\date{ }

\begin{document}
\sloppy
\maketitle
\begin{abstract}
A team of mobile agents, starting from distinct nodes of a network modeled as an undirected graph, have to meet at the same node and simultaneously declare that they all met. Agents execute the same algorithm, which they start when activated by an adversary or when an agent enter their initial node. While executing their algorithm, agents move from node to node by traversing edges of the network in synchronous rounds. Their perceptions and interactions are always strictly local: they have no visibility beyond their current node and can communicate only with agents occupying the same node. This task, known as \emph{gathering}, is one of the most fundamental problems in distributed mobile systems. Over the past decades, numerous gathering algorithms have been designed, with a particular focus on minimizing their time complexity, i.e., the worst-case number of rounds between the start of the earliest agent and the completion of the task. To solve gathering deterministically, a common widespread assumption is that each agent initially has an integer ID, called \emph{label}, only known to itself and that is distinct from those of all other agents. Labels play a crucial role in breaking possible symmetries, which, when left unresolved, may make gathering impossible. But must all labels be pairwise distinct to guarantee deterministic gathering?

In this paper, we conduct a deep investigation of this question by considering a context in which each agent applies a deterministic algorithm and has a label that may be shared with one or more other agents called \emph{homonyms}. A team $L$ of mobile agents, represented as the multiset of its labels, is said to be \emph{gatherable} if, for every possible initial setting of $L$, there exists an algorithm, even dedicated to that setting, that solves gathering. Our contribution is threefold. First, we give a full characterization of the gatherable teams. Second, we design an algorithm that gathers all of them in poly$(n,\log\lambda)$ time, where $n$ (resp. $\lambda$) is the order of the graph (resp. the smallest label in the team). This algorithm requires the agents to initially share only $O(\log \log \log \mu)$ bits of common knowledge, where $\mu$ is the multiplicity index of the team, i.e., the largest label multiplicity in $L$. Lastly, we show this dependency is almost optimal in the precise sense that no algorithm can gather every gatherable team in poly$(n,\log\lambda)$ time, with initially $o(\log \log \log \mu)$ bits of common knowledge.

As a by-product, we get the first deterministic poly$(n,\log\lambda)$-time algorithm that requires \emph{no} common knowledge to gather \emph{any} team in the classical case where all agent labels are pairwise distinct. While this was known to be achievable for teams of exactly two agents, extending it to teams of arbitrary size\textemdash under the same time and knowledge constraints\textemdash faced a major obstacle inherently absent in the two-agent scenario: that of termination detection. The synchronization techniques that enable us to overcome this obstacle may be of independent interest, as termination detection is a key issue in distributed systems.

\vspace{2ex}
\noindent {\bf Keywords:} gathering, deterministic algorithm, mobile agent.

\end{abstract}

\section{Introduction}

\subsection{Background}
Gathering a group of mobile agents is a fundamental problem that has received significant attention in the field of distributed mobile systems. This interest largely stems from the fact that gathering often serves as a critical building block for more complex collaborative tasks. Hence, understanding the theoretical limits and algorithmic foundations of gathering has direct implications for a broad range of distributed problems. 

Many studies assume that agents are equipped with unique
identifiers\textemdash often referred to as labels in the
literature\textemdash to help break symmetry and coordinate their
actions. However, this assumption may not always be realistic or
desirable. In practice, generating and maintaining globally unique
labels/IDs can be costly, especially in large-scale or dynamically
formed systems. Furthermore, in privacy-preserving scenarios, agents may deliberately blur their identity by revealing it only partially, thereby leading to the possible occurrence of homonyms.
Although probabilistic approaches can, in principle, bypass the need for labels, they typically fail to provide absolute guarantees in terms of complexity and/or accuracy of the resulting solutions, which may be an unacceptable compromise in safety-critical or mission-critical applications. 

This motivates a deep investigation of deterministic gathering in a more general context where some agents may share the same label, which we define as \emph{homonymous gathering}. Naturally, the presence of homonyms introduces new challenges in symmetry breaking and coordination, making gathering significantly harder. Our goal is to precisely characterize when gathering remains feasible in such a context, and whether it can be done efficiently.

\subsection{Model and Definitions}
\label{sec:model}
Given a team of $k\geq 2$ agents, an adversary selects a simple undirected connected graph $G=(V,E)$ with $n\geq k$ nodes and places the agents on $k$ distinct nodes of $G$. As is standard in the literature on gathering (cf. the survey \cite{Pelc19}), we make the following two assumptions regarding the labeling of the nodes and edges of $G$. The first assumption is that edges incident to a node $v$ are locally ordered with a fixed port numbering (chosen by the adversary) ranging from $0$ to ${\tt deg}(v)-1$ where ${\tt deg}(v)$ is the degree of $v$ (and thus, each edge has exactly two port numbers, one for each of its endpoints). The second assumption is that nodes are anonymous, i.e., they do not have any kind of labels or identifiers allowing them to be distinguished from one another. Given a node $v$ of $G$ and a port $i$ at this node, ${\tt succ}(v,i)$ will denote the node that is reached by taking port $i$ from node $v$. Given two adjacent nodes $v$ and $v'$ of $G$, ${\tt port}(v,v')$ will denote the port that must be taken from $v$ to reach $v'$.


Time is discretized into an infinite sequence of rounds $r=1,2,3,\ldots$. At the
beginning, every agent is said to be dormant and while an agent is dormant, it does
nothing. The adversary wakes up some of the agents at possibly different rounds. A dormant agent, not woken up by the adversary, is woken up
by the first agent that visits its starting node, if such an agent exists. Initially, each agent $X$ is provided with three elements that are known to $X$: a deterministic algorithm $\mathcal{A}$, some common knowledge $\mathcal{K}$ (detailed thereafter), and a label $\ell_X$ corresponding to a positive integer. The algorithm (resp. the common knowledge) is the same for all agents. However, for any two agents $X$ and $X'$, labels $\ell_X$ and $\ell_{X'}$ may be identical or not. When $\ell_X=\ell_{X'}$, we say that agents $X$ and $X'$ are \emph{homonyms}. 

In the rest of this paper, the team of mobile agents will be designated by the multiset $L$ of its agent's labels and $IS(L)$ will designate the initial setting that results from the choices of the adversary. Precisely, $IS(L)$ can be viewed as the underlying graph $G$ in which each node $v$ that is initially occupied by an agent $X$ is colored by a couple $(\ell_X,t)$ where $t$ is the round in which the adversary wants to wake up $X$ or $\perp$ if such a round does not exist. As mentioned above, common knowledge $\mathcal{K}$ is a piece of information, called \emph{advice}, that is initially given to all agents. Following the framework of algorithms with advice, see, e.g.,\cite{FraigniaudIP06,GlacetMP17,GorainMP23,MillerP15,NisseS09}, this information will correspond to a single binary string (possibly empty) that is computed by an oracle $\mathcal{O}$ knowing $IS(L)$. More formally, $\mathcal{O}$ is defined as a function that associates a string from $\{0,1\}^*$ to every possible initial setting of any team $L$. The length of the binary string initially provided by the oracle will be called the size of $\mathcal{K}$.

An agent starts executing algorithm $\mathcal{A}$ in the round of its wake-up. When executing algorithm $\mathcal{A}$ at a node $v$ in a round $r$, an agent $X$ first performs some computations depending on $\mathcal{K}$ and its memory in round $r$ that is denoted by $M_X(r)$ and that is formalized later (its memory includes its label $\ell_X$). Then, after the computations, agent $X$ chooses to do exactly one of the following things:

\begin{enumerate}
\item Declaring termination.
\item Waiting at its current node till the end of round $r$.
\item Moving to an adjacent node by a chosen port $p$.
\end{enumerate}

In the first case, the agent definitively stops the execution of $\mathcal{A}$ at node $v$ in round $r$. Otherwise, once the agent has processed the moving or waiting instruction, it is in round $r+1$ (at node $v$ in the second case, or at node ${\tt succ}(v,p)$ in the third case) and it continues the execution of algorithm $\mathcal{A}$. To ensure that the reference to the node occupied by an agent in any given round is unambiguous, we precisely assume that if the agent decides to take a port $p$ when located at node $v$ in round $r$, then it remains at node $v$ till the end of round $r$ and it is at node ${\tt succ}(v,p)$ at the beginning of round $r+1$.

When entering a node, an agent learns the port of entry. When located at a node $v$, a non-dormant agent sees the degree of $v$ and the other agents
located at that node, and it can access all information they currently hold.
(Agents' visions and interactions are thus always strictly local.) Also, if agents cross each other on an edge, traversing it simultaneously in different directions, they do not notice this fact.

An important notion used throughout this paper is the memory $M_X(r)$ of an agent $X$ in a given round $r$. Intuitively, it corresponds to all information it has collected until round $r$, both by navigating in the graph and by exchanging information with other agents (we assume that the agents have no limitations on the amount of memory they can use). We formalize $M_X(r)$ as a triple. Precisely, if, in round $r$, agent $X$ is dormant, or wakes up but is alone at its current node, then $M_X(r)$ is the triple $(\ell_X,\perp,\perp)$. Otherwise, we necessarily have $r>1$ and $M_X(r)$ is the triple $(M_X(r-1), A_{X}(r),B_X(r))$, where the terms $A_{X}(r)$ and $B_X(r)$ are defined as follows. Suppose that agent $X$ is at node $v$ in round $r$ and let $\{X_1,X_2,\ldots,X_h\}$ be the (possibly empty) set of the agents met by $X$ at $v$ in round $r$.

\begin{itemize}
\item If $X$ enters node $v$ in round $r$, then $X$ was at some node $u$ adjacent to $v$ in round $r-1$ and $A_{X}(r)=({\tt deg}(v),{\tt port}(u,v),{\tt port}(v,u))$. Otherwise $A_{X}(r)=({\tt deg}(v),\perp,\perp)$.
\item $B_X(r)$=$\emptyset$ if agent $X$ is alone in round $r$. Otherwise, $B_X(r)$=\\$\{(A_{X_1}(r),M_{X_1}(r-1)),(A_{X_2}(r),M_{X_2}(r-1)),\ldots,(A_{X_h}(r),M_{X_h}(r-1))\}$.
\end{itemize}


The rank of $M_X(r)$, denoted by {\tt rk}$(M_X(r))$, is equal to $0$ if $M_X(r)=(\ell_X,\perp,\perp)$, $1+{\tt rk}(M_X(r-1))$ otherwise.

Under the above-described context, the agents are assigned the task of {\em gathering}: there must exist a round in which the agents are at the same node, simultaneously declare termination and stop.

A team $L$ is said to be {\em gatherable} if for every initial setting $IS(L)$, there exists an algorithm $\mathcal{A}$, even specifically dedicated to $IS(L)$ (and thus regardless of $\mathcal{K}$) that allows the agents to solve gathering.\footnote{If $\mathcal{A}$ is specifically designed for $IS(L)$, then all necessary knowledge about $IS(L)$ can be ``embedded directly'' within the algorithm itself, making the use of $\mathcal{K}$ unnecessary in this case.} An algorithm $\mathcal{A}$ is said to be an {\em $\mathcal{O}$-universal gathering algorithm} if for every initial setting $IS(L)$ of any gatherable team $L$ it solves the task of gathering with $\mathcal{K}=\mathcal{O}(IS(L))$. The {\em time complexity of $\mathcal{A}$ with $\mathcal{O}$} (i.e., when the agents are initially provided with advice from oracle $\mathcal{O}$) is the worst-case number of rounds since the wake-up of the earliest agent until all agents declare termination.

We now finish the presentation of this section by giving some additional conventions. The smallest label in $L$ will be denoted by $\lambda$ and, for any label $\ell$, the length of its binary representation will be denoted by $|\ell|$. The number of occurrences of a given label in $L$ will be referred to as the multiplicity of that label.
Finally, the greatest common divisor of all the multiplicities of the labels in $L$ will be called the {\em symmetry index} of $L$, while the largest multiplicity of any label in $L$  will be called the {\em multiplicity index} of $L$.

\subsection{Related Work}
\label{sec:rela}
Historically, gathering was first investigated in the special case of exactly two agents, a variant known as \emph{rendezvous}. This foundational variant was popularized by Schelling in \emph{The Strategy of Conflict} \cite{Schelling}, a seminal book in game theory and international relations, where he examined how two individuals might coordinate to meet at a common location without prior communication. This led to the introduction of focal points—salient solutions toward which individuals naturally gravitate in the absence of explicit coordination. Since then, rendezvous and its generalization\textemdash gathering\textemdash have been widely studied across very diverse fields such as economics, sociology, and, most relevant to this work, distributed mobile systems. In such systems, gathering has emerged as a fundamental coordination task, especially in environments where agents operate under strict limitations on communication and perception. This is due to the fact that, in such environments, the ability to regroup may become a necessary first step toward enabling more sophisticated forms of collaboration and collective decision-making that would otherwise be infeasible.

Even when focusing exclusively on distributed mobile systems, the gathering problem has given rise to a vast amount of research. One reason is that the problem offers many meaningful parameters to explore: the type of environment in which the agents evolve, whether algorithms are deterministic or randomized, whether the agents move in synchronous rounds or not, etc. Since our work is naturally aligned with research on deterministic gathering in graphs where agents operate in synchronous rounds, we will concentrate essentially on this setting in the remainder of the subsection. However, for the reader interested in gaining broader insight into the problem, we refer to \cite{AgmonP06,CieliebakFPS12,CzyzowiczPL12,IzumiSKIDWY12}, which investigate asynchronous gathering in the plane or in graphs. Regarding randomized rendezvous, a good starting point is to go through \cite{Alpern02,Alpern03,KranakisKR06}.

The literature on deterministic gathering in graphs, as typically studied in the synchronous setting and thoroughly surveyed in~\cite{Pelc19}, spans a wide range of variants shaped by the capabilities granted to the agents. Two particularly influential dimensions concern the ability to leave traces of their passage (e.g., by dropping markers or writing on local node memories) and the extent of their visibility, which may range from a global view of the network to a strictly local view limited to their current node. Perhaps the most compelling case arises when both dimensions are taken in their minimal forms\textemdash namely, no traces and node-local view only. Within this constrained and challenging context, a substantial body of work has focused on minimizing the time (i.e., the number of rounds) to achieve gathering. The core model adopted in most of these contributions is the one we formalize in Section~\ref{sec:model}, restricted to the special case where all agents have distinct labels (to break potential symmetries)\textemdash though sometimes with a variation in how agents are introduced into the network. Specifically, the literature considered two main scenarios in this regard: a \emph{grounded scenario} and a \emph{dropped-in scenario}. In the grounded scenario, all agents are initially present in the network, but start executing the algorithm when activated by an adversary or as soon as an agent enters their initial node (this potentially allows the algorithm to exert some control over the delays between agents' starting times). By contrast, in the dropped-in scenario, agents are not necessarily present at the outset: each of them appears in the network at a time and location entirely determined by an adversary and begins its execution in the round of its arrival. In this latter case, the delays between the starting times cannot be influenced by the algorithm executed by the agents.

Early works primarily focused on the dropped-in scenario, among which the key contributions are~\cite{DessmarkFKP06,KowalskiM08,Ta-ShmaZ14}. In \cite{DessmarkFKP06}, the authors provide a rendezvous algorithm (thus working for two agents only) whose time complexity is polynomial in the order $n$ of the graph, in $|\lambda|$\textemdash where $\lambda$ is the smaller of the two labels\textemdash and in the delay $\tau$ between the agents' starting times, even when time is counted from the start of the later agent. While a dependence on $\tau$ is naturally expected in the dropped-in scenario when measuring time from the start of the earlier agent (as the second agent's arrival is beyond the algorithm's control), it was unclear whether such a dependence remained unavoidable when counting time from the arrival of the later agent. This led the authors of \cite{DessmarkFKP06} to ask whether a polynomial-time solution depending only on $n$ and $\lambda$ could be achieved in that case. In \cite{KowalskiM08}, a positive answer is brought to this question, with an algorithm ensuring rendezvous in time $\mathcal{O}(|\lambda|^3 + \mbox{~~}n^{15}\log^{12} n)$ from the start of the later agent. This was subsequently improved in \cite{Ta-ShmaZ14},
which reduced the complexity to $\tilde{\mathcal{O}}(n^5 |\lambda|)$. On the negative side,~\cite{DessmarkFKP06} proves that any rendezvous algorithm necessarily has time complexity at least polynomial in $n$ and $|\lambda|$ (specifically, at least $\Omega(n |\lambda|)$), even when $\tau = 0$.

To deal with more than two agents, the authors of \cite{KowalskiM08} observe that it is sufficient for agents to execute a rendezvous algorithm with the following simple adaptation called the \emph{sticking-together strategy}: whenever two or more agents meet, the agent with the smallest label continues its execution as if nothing had happened, while the others start following it, mimicking its actions from that point on. This way the agents will indeed eventually end up together. However, detecting all the agents are together in the dropped-in scenario requires the agents to know the exact number of agents that will be introduced by the adversary, regardless of the strategy being employed. This \emph{de facto} rules out the possibility of designing gathering algorithms using no common knowledge, unless one entirely gives up on termination detection\textemdash a significant limitation in distributed systems. This fundamental barrier disappears in the grounded scenario, which has emerged as the prevailing one for studying gathering, as reflected in the large body of work dedicated to it (see, e.g.,~\cite{BouchardDD16,BouchardDL22,BouchardDP23,DieudonnePP12,HiroseNOI22,HiroseNOI24,MollaMM23,Pelc18,SaxenaM24}). Here, gathering \emph{any} number of agents without relying on any common knowledge becomes possible (see, e.g.,~\cite{BouchardDP23}). Moreover, since the delays between agents' starting times can now be algorithmically controlled to some extent, the reference time for rendezvous and gathering is taken from the start of the earliest agent. Interestingly, in this new reference time, the poly$(n,|\lambda|)$-time lower bound established by~\cite{DessmarkFKP06} in the dropped-in scenario remain valid in the grounded context, since it was proven even when  $\tau = 0$. Yet, within this same context, and despite the numerous studies of the literature, no algorithm was known to guarantee gathering in polynomial time in $n$ and $|\lambda|$, without assuming any common knowledge, except in the special case where teams are restricted to exactly two agents (cf. \cite{BouchardDPP18}).\footnote{More precisely, the authors of~\cite{BouchardDPP18} analyze rendezvous under the assumption that agents may traverse edges at different speeds, and design an algorithm running in time polynomial in the order of the graph, the logarithm of the smaller label, and the maximum edge traversal time. In the particular case where every edge traversal time equals~1, this yields a rendezvous algorithm polynomial in the first two parameters.} In fact, the only existing gathering algorithms achieving such a complexity for teams of arbitrary size required the knowledge of a polynomial upper bound on the graph order (see~\cite{BouchardDL22,BouchardDP23}).

Another intriguing aspect of the literature concerns the assumptions
made about the agents' labeling. Specifically, gathering has only been
studied at two opposite ends of the spectrum: either when all agents
have pairwise distinct labels (as in the papers cited in the previous
paragraph), or, much less frequently, when all are completely
anonymous, equivalent to all agents sharing the same label (as
in~\cite{DieudonneP16,PelcY19}). Labels play a crucial role in
breaking potential symmetries, which, if left unresolved, can often
render gathering impossible, particularly in the anonymous
case. However, requiring all labels to be distinct may be
unnecessarily strong. Thus, the question of
what can be achieved for gathering under such an intermediate
assumption, where some degree of homonymy is allowed, remains entirely
open.

To conclude, it is worth noting that the impact of homonymy has already been investigated, yet in the very different context of static message-passing networks, for classical problems such as leader election and sorting \cite{YK89c,FKKLS04j,CGM12j,DP16j}.

\subsection{Our Results}

In this paper, we conduct a thorough study of the gathering problem in a standard model from the literature, which we extend to allow the presence of homonyms. Actually, in light of the discussions in the related work section, we are particularly motivated by the following open questions:

\begin{itemize}
\item Under what conditions is a team of agents\textemdash with possible homonyms\textemdash gatherable?

\item Does there exist an algorithm allowing to gather any gatherable team in poly$(n,|\lambda|)$-time, without any initial common knowledge?

\item And if not, what is the minimum amount of common knowledge required?
\end{itemize}

With these questions in mind, we make the following three contributions. First, we give a full characterization of the gatherable teams. Second, we design an algorithm that gathers all of them in poly$(n,|\lambda|)$ time. This algorithm requires the agents to initially share only $O(\log \log \log \mu)$ bits of common knowledge, where $\mu$ is the multiplicity index of the team. And finally, we show this dependency is almost optimal in the precise sense that no algorithm can gather every gatherable team in poly$(n,|\lambda|)$ time, with initially $o(\log \log \log \mu)$ bits of common knowledge.

As an important by-product, we also obtain the first poly$(n,|\lambda|)$-time algorithm that requires \emph{no} common knowledge to gather \emph{any} team in the classical scenario where all agent labels are pairwise distinct. While this result has been known to be achievable for the special case of exactly two agents, extending it to teams of arbitrary size has remained a fundamental open problem until now (cf. Section~\ref{sec:rela}). The core difficulty of the extension lies in detecting the termination of the task. In a two-agent context, termination detection is trivial, as it inherently occurs at the time of the first meeting. However, when dealing with teams of arbitrary size, this is no longer the case: it becomes a fundamental challenge within the constraints of poly$(n,|\lambda|)$ time and no common knowledge (especially without any upper bound on the team size). Our ability to overcome this challenge is a direct consequence of the synchronization techniques we developed, in the general case with homonyms, to make the agents determine whether the gathering is done or not. These techniques are likely to be of independent interest, as termination detection is a key issue in distributed systems.

\section{Preliminaries}
\label{sec:preli}
In this section, we introduce some additional definitions and basic procedures that will be used throughout this paper.

Let us start by introducing the order $\prec$ on the set $\mathcal{M}$ of all possible memories. In particular, this order is made to have the heredity property stated in Lemma~\ref{lem:prec}. Let $f$ be any injective function from $\mathcal{M}$ to $\mathbb{N}$ (such a function necessarily exists since $\mathcal{M}$ is recursively enumerable). Given two memories $M_1=(M'_1,*,*)$ and $M_2=(M'_2,*,*)$, we have $M_1 \prec M_2$ if: 

\begin{itemize}
\item ${\tt rk}(M_1) < {\tt rk}(M_2)$; 
\item Or ${\tt rk}(M_1) = {\tt rk}(M_2)=0$ and $f(M_1)<f(M_2)$;  
\item Or ${\tt rk}(M_1) = {\tt rk}(M_2)>0$, $M'_1=M'_2$ (resp. $M'_1\ne M'_2$) and $f(M_1)<f(M_2)$ (resp. $M'_1\prec M'_2$).
\end{itemize}

Note that $\prec$ induces a strict total order on $\mathcal{M}$. 

When $M_1 \prec M_2$ (resp. $M_2 \prec M_1$), we will say that $M_1$ is \emph{smaller} (resp. \emph{larger}) than $M_2$. Below is a lemma related to memories.

\begin{lemma}
\label{lem:prec}
Let $A$ and $B$ be two agents that are non dormant in a round $r$. If $M_A(r)\prec M_B(r)$, then $M_A(r')\prec M_B(r')$ in every round $r' \geq r$.
\end{lemma}
\begin{proof} 
  Assume two agents $A$ and $B$  that are non dormant in a round $r$ and $M_A(r)\prec M_B(r)$. We now show by induction on $r'$ that  $M_A(r')\prec M_B(r')$ in every round $r' \geq r$. The base case $r'=r$ is trivial.

  Assume now that $r' > r$. Since $A$ and $B$ are non dormant in a
  round $r$ and $r' > r$, we have both ${\tt rk}(M_A(r')) > 0$ and ${\tt rk}(M_B(r')) > 0$. So, we can let $M_A(r')=(M_A(r'-1),*,*)$ and
  $M_B(r')=(M_B(r'-1),*,*)$. Then, we have two cases:
  \begin{itemize}
  \item ${\tt rk}(M_A(r')) \neq {\tt rk}(M_B(r'))$.

    So, ${\tt rk}(M_A(r)) \neq {\tt rk}(M_B(r))$ too. Now, in this case, $M_A(r)\prec M_B(r)$ implies that ${\tt rk}(M_A(r)) < {\tt rk}(M_B(r))$, which in turn implies ${\tt rk}(M_A(r')) < {\tt rk}(M_B(r'))$ and thus $M_A(r')\prec M_B(r')$.

  \item ${\tt rk}(M_A(r')) = {\tt rk}(M_B(r')) = k$. We already know that
    $k > 0$. Moreover, by induction
    hypothesis, $M_A(r'-1)\prec M_B(r'-1)$, which in particular means that
    $M_A(r'-1) \neq M_B(r'-1)$ (indeed, $\prec$ is a strict total
    order). We have then $M_A(r')\prec M_B(r')$ by definition of
    $\prec$.
    \end{itemize}
 Hence, the induction holds and the lemma follows. 
\end{proof}

A sequence $\pi$ of $2k$ non-negative integers
$(y_1,y_2,\ldots,y_{2k})$ is said to be a {\em path} from a node $u$
iff (1) $k=0$ or (2) $k\geq 1$, $y_1 < {\tt deg}(u)$, ${\tt port}({\tt
  succ}(u,y_1),u)=y_2$ and $(y_3,\ldots,y_{2k})$ is a path from node
${\tt succ}(u,y_1)$. The length of $\pi$, corresponding to its number
of integers, will be denoted by $|\pi|$, and its i$th$ integer will be
referred to as $\pi[i]$. For every integer $0\leq i \leq k$,
$\pi[1,2i]$ will denote the path corresponding to the prefix of $\pi$
of length $2i$. Finally, the concatenation of the path $\pi_1$ from
$u$ to $v$ and the path $\pi_2$ from $v$ to $w$ will be a path from
$u$ to $w$ denoted by $\pi_1\pi_2$. This notion of path in particular
used in the proof of the next lemma.

\begin{lemma}
\label{lem:diff}
Let $A$ and $B$ be two agents sharing the same node $v$ in a round $r$. We have $M_A(r)\ne M_B(r)$.
\end{lemma}

\begin{proof}
Since $A$ and $B$ are initially placed at different nodes, $r > 1$ and
both robots are awake at round $r$. Let $P_A$ (resp.  $P_B$) the path followed by $A$ (resp. $B$)
since its wake-up.
Since $A$ and $B$ are located at the same node at round $r$ but
initially placed at different nodes, $P_A$ necessarily differs from
$P_B$, which implies that $M_A(r)\ne M_B(r)$.
\end{proof}

We will use $\log$ to denote the binary logarithm. We will say that an execution $\mathcal{E}$ of a sequence of instructions lasts $T$ rounds if the number of edges it prescribes to traverse plus the number of rounds it prescribes to wait is equal to $T$. Moreover, if $\mathcal{E}$ starts in a round $r$, we will say that $\mathcal{E}$ \emph{is completed in} (resp. \emph{is completed by}) round $r+T$ if $\mathcal{E}$ lasts exactly $T$ rounds (resp. at most $T$ rounds).


We now present three procedures that will be used as building blocks to design our algorithm given in Section~\ref{sec:pos}. They require no common knowledge $\mathcal{K}$.

The first procedure, due to Ta-Shma and Zwick \cite{Ta-ShmaZ14}, allows to gather a team made of exactly two agents $X$ and $X'$ having distinct labels. We will call it ${\tt TZ}(\ell)$, where $\ell$ is the label of the executing agent. In \cite{Ta-ShmaZ14}, the authors show that if $X$ and $X'$ start executing procedure ${\tt TZ}$ in possibly different rounds, then they will meet after at most $(n\cdot\min(|\ell_X|,|\ell_{X'}|))^{\alpha}$ rounds since the start of the later agent, for some integer constant $\alpha\geq 2$. 


The second procedure is based on universal exploration sequences (UXS)
and is a corollary of the result of Reingold \cite{Reingold08}. It is
denoted by ${\tt EXPLO}(N)$, where $N$ is any integer such that
$N\geq2$. The execution of ${\tt EXPLO}(N)$ by some agent $X$ makes it
traverse edges of the underlying graph $G$, without any waiting period. This execution terminates
within at most $N^{\beta}$ rounds, where $\beta\geq 2$ is some
constant integer. Moreover, if $N$ is an upper bound on $n$ (i.e., the
number of nodes in $G$), we have the guarantee that each node of $G$
will be visited at least once by $X$ during its execution of ${\tt
  EXPLO}(N)$, regardless of its starting node.


The third procedure, explicitly described in \cite{BouchardDD16} (and
derived from a proof given in \cite{ChalopinDK10}), allows an agent to
visit all nodes of $G$ assuming there is a fixed token at its starting
node and no token anywhere else. As our model does not assume the
existence of tokens, this routine will be later adapted to our needs
by making certain agents play the role of token. We call this
procedure ${\tt EST}$, for \emph{exploration with a stationary
token}. The execution of ${\tt EST}$ in $G$ lasts at most $8n^5$
rounds and at the end of the execution the exploring agent is back to its
starting node. Since this will be important for our purposes, we give
below an almost verbatim description of this procedure given in
\cite{BouchardDD16}.


When executing procedure ${\tt EST}$, the agent actually constructs
and stores a {\em port-labeled $BFS$ tree} $T$ which is an appropriate
representation of a BFS spanning tree of the underlying graph $G$. By
appropriate, we mean a representation that enables the agent to traverse an
actual BFS spanning tree of $G$.  In more detail, there will be a correspondance
between the nodes of $T$ and the nodes of $G$: we will denote by
$c(u)$ the node of $T$ corresponding to the node $u$ of $G$.
$T$ will be rooted at node $c(v)$ such that $v$ is the node where the
agent is located at the beginning of the procedure ${\tt EST}$.
Moreover, we will ensure that the incoming and outgoing ports along
the simple path from the root $c(v)$ to any node $c(u)$ in $T$ will
describe a shortest path from $v$ to $u$ in $G$.


Nodes, edges and port numbers are added to the port-labeled $BFS$ tree as
follows. At the beginning, the agent is at node $v$: it initializes
$T$ to the root node $c(v)$ and then makes the {\em process} of
$c(v)$. The process of any node $c(w)$ of $T$ consists in checking
each neighbor $x$ of $w$ in order to determine whether a corresponding
node $c(x)$ has to be added to the tree or not. When starting the
process of $c(w)$, the agent, which is located at $w$, takes port $0$
to move to the neighbor ${\tt succ}(w,0)$ and then checks this
neighbor.

When a neighbor $x$ of $w$ is being checked, the agent verifies
whether a node corresponding to $x$ has been previously added to
$T$. To do this, for each node $c$ of $T$, the agent considers the
simple path $P$ from node $c$ to the root $c(v)$ in $T$ and then,
starting from node $x$, tries to successively take the following
$\frac{|P|}{2}$ ports in the underlying graph $G$:
$P[1],P[3],P[5]\ldots,P[|P|-1]$. The agent gets the guarantee that $c$
cannot correspond to $x$ if after taking $0\leq i \leq \frac{|P|}{2}$
ports, one of these three conditions is satisfied:

\begin{enumerate}
\item $0<i$ and the port by which the agent enters its current node is not $P[2i]$.
\item $i<\frac{|P|}{2}$ and there is no port $P[2i+1]$ at its current node. 
\item $i=\frac{|P|}{2}$ and the current node does not contain the token.
\end{enumerate}

If each of these conditions is never satisfied, then the agent leaves
(resp. enters) successively by ports $P[1],P[3],P[5]\ldots,P[|P|-1]$
(resp. $P[2],P[4],P[6]\ldots,P[|P|]$) and then meets the token: the
agent thus gets the guarantee that $c$ corresponds to $x$.


As soon as the agent gets the guarantee that $c$ corresponds to $x$ or
not (after taking $0\leq i \leq \frac{|P|}{2}$ ports), it goes back to
node $x$. If $i>0$, it does so by successively taking ports
$P[2i],P[2i-2],P[2i-4]\ldots,P[2]$; otherwise it is already at $x$.
After that, if $c$ does not correspond to $x$, but it remains at least
one node $c'$ of $T$ for which the agent does not know yet whether $c'$
corresponds to $x$ or not: it thus repeats with $c'$ what it has done
for $c$.  Eventually, either the agent determines that some node $c$ of $T$
corresponds to $x$ and $x$ is \emph{rejected}, or no node of $T$
corresponds to $x$ and $x$ is \emph{admitted}.  If $x$ ends up to be
admitted, a corresponding node $c(x)$ is attached to $c(w)$ in $T$ by
an edge whose port number at $c(w)$ (resp. $c(x)$) is equal to ${\tt
  port}(w,x)$ (resp. ${\tt port}(x,w)$).


Once node $x$ is admitted or rejected, the agent goes back to $w$ by
taking ${\tt port}(x,w)$ and then moves to an unchecked neighbor of
$w$, if any, following the increasing order of w's port numbers. When
all neighbors of $w$ have been checked, the process of $c(w)$ is
completed and the agent goes back to $v$ using the simple path from
$c(w)$ to $c(v)$ in $T$. Let $\mathcal{Y}$ be the set of all paths in
the port-labeled $BFS$ tree from its root to a node that has not yet
been processed. If $\mathcal{Y}=\emptyset$ then procedure ${\tt EST}$
is completed (and the port-labeled $BFS$ tree spans the underlying
graph). Otherwise, the agent considers the lexicographically smallest
path of $\mathcal{Y}$, which leads from the root $c(v)$ to some
unprocessed node $c(s)$ in $T$, and follows this path in $G$ to get at
node $s$. When getting at this node, the agent starts the process of
$c(s)$.

\section{Characterization of the Gatherable Teams}
\label{sec:charac}

The aim of this section is to prove Theorem~\ref{theo:charac}, which provides a full characterization of the gatherable teams.

\begin{theorem}
\label{theo:charac}
A team $L$ is gatherable if and only if its symmetry index is $1$.
\end{theorem}

The proof of Theorem~\ref{theo:charac} follows directly from Lemmas~\ref{lem:charac2} and~\ref{lem:charac1} given below.

\begin{lemma}
\label{lem:charac2}
If the symmetry index of a team $L$ is $1$, then $L$ is gatherable.
\end{lemma}

\begin{proof}
Consider any team $L$ of $k$ agents for which the symmetry index is $1$. By definition, $L$ is gatherable if for every initial setting $IS(L)$, there exists an algorithm even specifically dedicated to $IS(L)$ that allows the agents to solve gathering. Hence, to show that the lemma holds, it is enough to show that Algorithm $\mathcal{A}$, whose pseudocode is given below in Algorithm~\ref{alg:gather_specific}, allows to solve the task of gathering, assuming that the multiset $L$ and the graph order $n$ are provided as input. In the rest of this proof, we will denote by $r_0$ the first round in which an agent wakes up and we will say that an agent is \emph{advanced} if it has at least completed the execution of line~\ref{z1} of Algorithm~\ref{alg:gather_specific}.

Besides Procedure ${\tt EXPLO}$ introduced in Section~\ref{sec:preli}, the pseudocode of Algorithm~\ref{alg:gather_specific} involves three particular functions. The first one is ${\tt CurMultLab}()$: when an agent $X$ calls it from a node $u$, the function simply returns the multiset of labels of all advanced agents that are currently at node $u$ (including $X$'s label if $X$ is advanced). The second function is any fixed injective map $f:\mathscr{P}(L) \rightarrow \left[1 \ldotp\ldotp 2^k \right]$, where $\mathscr{P}(L)$ is the set of multisubsets of $L$ (such function exists as the cardinality of $\mathscr{P}(L)$ is at most $2^k$). The third function is ${\tt TZ}(\ell, i)$, where $i$ and $\ell$ are positive integers. When calling ${\tt TZ}(\ell, i)$ in a round $r$, the function returns the $i$th action (move through some port $p$ or wait during the current round) that would be asked by Procedure ${\tt TZ}(\ell)$ in an execution $\mathcal{E}$ of the procedure starting from the node occupied by the agent in round $r-i+1$, assuming the agent is alone in the graph (in which case, $\mathcal{E}$ never terminates). Note that ${\tt TZ}(\ell, i)$ is indeed computable by the agent in round $r$, because according to Algorithm~\ref{alg:gather_specific}, if $i\geq2$, then, in round $r-1$, the agent called ${\tt TZ}(\ell, i-1)$ and performed the $(i-1)$th action of $\mathcal{E}$. 

As stated in Section~\ref{sec:preli}, Procedures ${\tt TZ}$ and ${\tt EXPLO}$ satisfy the following properties in a graph of order $n$. First, if two agents $X$ and $X'$ with distinct labels $\ell_X$ and $\ell_{X'}$ start executing, in possibly different rounds, ${\tt TZ}(\ell_X)$ and ${\tt TZ}(\ell_{X'})$, then they will meet after at most $(n\cdot\min(|\ell_X|,|\ell_{X'}|))^{\alpha}$ rounds since the start of the later agent, for some constant $\alpha\geq 2$. Second, an execution of ${\tt EXPLO}(N)$ by an agent $X$ lasts at most $N^{\beta}$ rounds and allows $X$ to visit every node of the graph if $n\leq N$. In particular, the latter property implies that every agent becomes advanced by round $r_0+2n^{\beta}$. 

\begin{algorithm}[H]
  \label{alg:gather_specific}
  \caption{Algorithm $\mathcal{A}(L,n)$}
  \SetNoFillComment
  perform ${\tt EXPLO}(n)$\;\label{z1}
  $S := {\tt CurMultLab}()$\;
  $i :=1$\;
  \While{$S\neq L$}{
    perform ${\tt TZ}(f(S),i)$\;
    \eIf{$S\neq {\tt CurMultLab}()$}{
      $S := {\tt CurMultLab}()$\;
      $i :=1$\; 
    }
    {
      $i :=i+1$\;
    }
  }
  declare termination\;
  \end{algorithm}

The execution of the while loop of Algorithm $\mathcal{A}$ by an agent $X$ can be viewed as a sequence of phases. In each phase, agent $X$ executes step by step Procedure ${\tt TZ}(f(S))$ where $S$ is the multiset of labels of the advanced agents that are at $X$'s node at the beginning of the phase. A phase ends as soon as the multiset of labels of the advanced agents that are currently at $X$'s node changes: if the new multiset of advanced agents' labels $S'$ is $L$, then $X$ declares termination (and so does every other agent), otherwise it begins a new phase by restarting from scratch a step-by-step execution of Procedure ${\tt TZ}$ but with $f(S')$, instead of $f(S)$, as input. 

Note that agent $X$ declares termination in some round $r$, if, and only if, in this round the multiset of labels of all advanced agents that are at its current node corresponds to $L$ (in which case, all agents also declare termination in round $r$). Also note that according to Algorithm $\mathcal{A}$, when two advanced agents meet, they stay together forever thereafter as they will necessarily perform the exact same actions from this event onward. Thus, since every agent becomes advanced by round $r_0+2n^{\beta}$, to prove the lemma, it is enough to prove the following: for any round $r'\geq r_0+2n^{\beta}$, if not all agents are gathered in round $r'$, then there exist two agents that are not together in round $r'$ that meet in some round $r''>r'$. Suppose by contradiction it is not the case for some round $r'$. 

Since the symmetry index of $L$ is $1$, there are two distinct nodes $u$ and $u'$ that are occupied in round $r'$ by some agent $X$ and some agent $X'$, respectively, and such that the multiset of labels $S$ of all agents that are at $u$ in $r'$ is different from the multiset of labels $S'$ of all agents that are at $v$ in $r'$. From the previous explanations, we know that $X$ and $X'$ cannot be together in any round of $[r_0+2n^{\beta} .. r']$; moreover, the contradictive assumption implies that $X$ and $X'$ cannot meet after round $r'$. In other words, $X$ and $X'$ can never share the same node from round $r_0+2n^{\beta}$ on.

In view of Algorithm $\mathcal{A}$, the contradictive assumption and the fact that every agent becomes advanced by round $r_0+2n^{\beta}$, there is a round $r_0+2n^{\beta}\leq t\leq r'$ (resp. $r_0+2n^{\beta}\leq t'\leq r'$) such that in round $r'+(2^kn)^{\alpha}$, agent $X$ (resp. $X'$) has just finished the first $(2^kn)^{\alpha}+r'-t$ (resp. $(2^kn)^{\alpha}+r'-t'$) actions of some execution of ${\tt TZ}(f(S))$ (resp. ${\tt TZ}(f(S'))$). Assume without loss of generality that $t\leq t'$. It follows that agents $X$ and $X'$ necessarily meet in a round $t'\leq t''\leq r'+(2^kn)^{\alpha}$ in view of the property of Procedure ${\tt TZ}$ recalled above and the fact that $1\leq f(S)<f(S')\leq 2^k$ or $1\leq f(S')<f(S)\leq 2^k$. However, since $r_0+2n^{\beta}\leq t'$, we obtain a contradiction with the fact that $X$ and $X'$ can never share the same node from round $r_0+2n^{\beta}$ on. Hence, the lemma follows.
\end{proof}

\begin{lemma}
\label{lem:charac1}
If the symmetry index of a team $L$ is not $1$, then $L$ is not gatherable.
\end{lemma}

\begin{proof}
Let $L$ be a team of $k$ agents with symmetry index $\sigma \neq 1$, let $m_L(\ell)$ denote the multiplicity of label $\ell$ in $L$, and let $\Lambda$ be the support of $L$. We construct an initial configuration $IS(L)$ as follows.
Let the underlying graph be a ring $G$ of size $k\sigma$, with nodes $v_0, v_1, \ldots, v_{k\sigma-1}$ arranged in clockwise order. Each node $v_i$ has two ports: port 0 leads to $v_{i+1}$, and port 1 leads to $v_{i-1}$ (indices are taken modulo $k\sigma$). Let $h: \{0, 1, \dots, \frac{k}{\sigma}-1\} \rightarrow \Lambda$ be any map such that $|h^{-1}(\ell)| = \frac{m_L(\ell)}{\sigma}$ for all $\ell \in \Lambda$ (such a map exists as $\frac{\sum_{\ell\in\Lambda}m_L(\ell)}{\sigma}=\frac{k}{\sigma}$).
Now, place agents on the ring as follows: for each $i \in \{0, \ldots, \frac{k}{\sigma}-1\}$ and each $j\in \{0, \ldots, \sigma-1\}$, place an agent with label $h(i)$ at node $v_{i\sigma+jk}$. The other nodes, i.e., the nodes $v_i$ such that $i\not\equiv 0\mod \sigma$, do not initially contain any agent. Observe that for each $j\in \{0, \ldots, \sigma-1\}$ the set of nodes $\{v_i\ |\ jk\leq i\leq (j+1)k-1\}$ contains exactly $\frac{m_L(\ell)}{\sigma}$ agents with label $\ell$. It follows that the agents placed in $G$ correspond exactly to the team $L$. We assume that all agents are activated in round 1. See Figure~\ref{fig:lemma_3_1_example} for an example of such a construction. 

\begin{figure}
  \begin{center}
      \includegraphics[width=9cm]{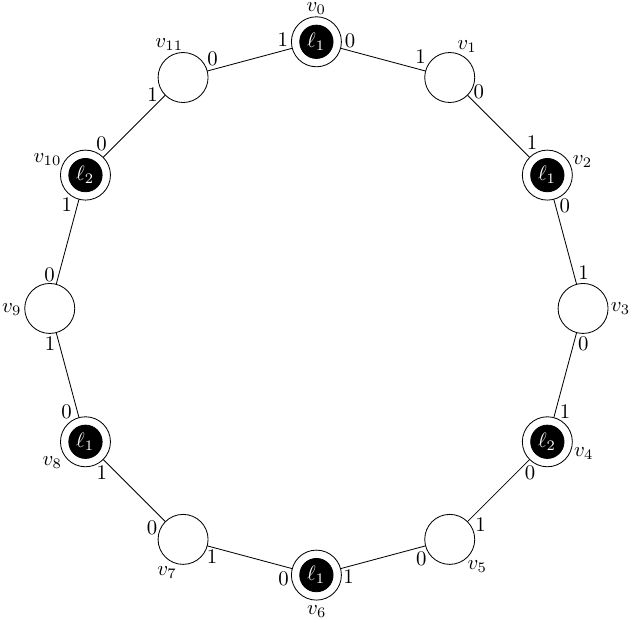}
  \caption{Initial configuration for a team of agents $\{\ell_1,\ell_1,\ell_1,\ell_1,\ell_2,\ell_2\}$ with $k=6$, $\sigma =2$ and with fixed map $h$ such that $h(0)=\ell_1$, $h(1)=\ell_1$ and $h(2)=\ell_2$. Agents are represented as black disk with their label in white. \label{fig:lemma_3_1_example}}
  \end{center}
\end{figure}

Let $S_i^r$ denote the set of agents located at node $v_i$ in round $r$. We will prove by induction that for any deterministic algorithm $\mathcal{A}$, the following property holds for all rounds $r \in \mathbb{N}^*$:

$H(r)$: For all $0 \leq i<j \leq k\sigma-1$ such that $i \equiv j \mod k$, there exists a bijection
$f_{i,j}^r : S_i^r \to S_j^r$ such that for every $X \in S_i^r$, the memory of $X$ in round $r$ equals that of $f_{i,j}^r(X)$.

First, we show the base case $H(1)$ of the induction.  Let $i, j \in$ $\{0, \dots, k\sigma-1\}$ with $i<j$ and $i \equiv j \equiv m \mod k$ for some $m$ such that $0 \leq m < k$. If $m \not\equiv 0\mod \sigma$, then $v_i$ and $v_j$ are both empty, and the bijection $f_{i,j}^1$ is trivial. Hence, we can assume that $m \equiv 0\mod \sigma$. This means that $v_i$ and $v_j$ each host a single agent with label $h(m/\sigma)$. In round $1$, each agent has initial memory $(\ell, \bot, \bot)$ where $\ell$ is its label. Therefore, for any such $i, j$, the singleton sets $S_i^1$ and $S_j^1$ contain agents with identical memories. The bijection $f_{i,j}^1$ is also trivial in this case, and thus $H(1)$ necessarily holds.

Assume $H(r)$ holds for some round $r \geq 1$. We show that $H(r+1)$ also holds. Denote by $\pi(\cdot)$ the reduction modulo $k\sigma$.
For any node $v_i$, let us partition $S_i^{r+1}$ into three disjoint sets:

\begin{itemize}
  \item $\dot{S}_i^{r+1}$: agents that decided in round $r$ to stay idle at $v_i$,
  \item $\vec{S}_i^{r+1}$: agents that decided in round $r$ to move to $v_i$ from $v_{\pi(i-1)}$ via port $0$,
  \item $\cev{S}_i^{r+1}$: agents that decided in round $r$ to move to $v_i$ from $v_{\pi(i+1)}$ via port $1$.
\end{itemize}

By the inductive hypothesis $H(r)$, for all $i \equiv j \mod k$ such that $i<j$, there exist bijections $f^r_{i,j}$, $f^r_{\pi(i-1),\pi(j-1)}$ and $f^r_{\pi(i+1),\pi(j+1)}$ from $S_i^r$ to $S_j^r$, $S_{\pi(i-1)}^r$ to $S_{\pi(j-1)}^r$ and $S_{\pi(i+1)}^r$ to $S_{\pi(j+1)}^r$, respectively, that all preserve the agents' memories in round $r$. Since the algorithm is deterministic, agents with the same memories take the same action. Hence, there exists a bijection from $\dot{S}_i^{r+1}$ to $\dot{S}_j^{r+1}$ that preserves the memories hold in round $r$.
Similarly, such bijections exist between $\vec{S}_i^{r+1}$ and $\vec{S}_j^{r+1}$, and between $\cev{S}_i^{r+1}$ and $\cev{S}_j^{r+1}$.
From these, we can construct a global bijection $f_{i,j}^{r+1} : S_i^{r+1} \to S_j^{r+1}$
that preserves both the memories hold in round $r$ and the actions taken in round $r$.
Let us now verify that this bijection necessarily preserves the agents' memories hold in round $r+1$. For any agent $X \in S_i^{r+1}$, define $X' = f_{i,j}^{r+1}(X)$. Their memories  in round $r+1$ are respectively 
$M_X(r+1) = (M_X(r), A_X(r+1), B_X(r+1))$ and $M_{X'}(r+1) = (M_{X'}(r), A_{X'}(r+1), B_{X'}(r+1))$.
Since $f_{i,j}^{r+1}$ preserves the memories hold in round $r$ (resp. the actions taken in round $r$), we have $M_X(r) = M_{X'}(r)$ (resp. $A_X(r+1) = A_{X'}(r+1)$). The set $B_X(r+1)$ is also preserved because $B_X(r+1) = \{(A_Y(r+1), M_Y(r)) \mid Y \in S_i^{r+1} \setminus \{X\}\}$, $B_{X'}(r+1)=\{(A_{Y'}(r+1), M_{Y'}(r)) \mid Y' \in S_j^{r+1} \setminus \{X'\}\}$ and $f_{i,j}^{r+1}$ preserves the memories hold in round $r$ as well as the actions taken in round $r$.
Hence, $M_X(r+1) = M_{X'}(r+1)$, completing the inductive step. Therefore, $H(r)$ holds for all $r \geq 1$ by induction. 

Consequently, in any round $r$, for any agent $X$ at node $v_i$, there exists a corresponding agent $f^{r}_{i,j}(X)$ at node $v_{j}$ for some $j\ne i$ such that $j \equiv i \mod k$. It follows that gathering is never achieved.
\end{proof}

\section{A Polynomial-Time Algorithm Using Common Knowledge of Small Size}
\label{sec:pos}

In this section, we present an algorithm, called ${\mathcal{HG}}$ (short for {\it Homonymous Gathering}). Essentially, this algorithm enables the gathering of any gatherable team $L$, with multiplicity index $\mu$, in a time at most polynomial in $n$ and $|\lambda|$, provided that $\mathcal{K}$ is the binary representation of $\lceil\log\log (\mu+1)\rceil$. Note that the size of $\mathcal{K}$ is thus in ${O}(\log \log \log \mu)$.


The remainder of this section is structured as follows. First, we provide the intuition underlying Algorithm ${\mathcal{HG}}$. Next, we present a detailed description of the algorithm. Finally, we establish its correctness and analyze its time complexity.

\subsection{Intuition}

At first glance, it could be quite tempting to simply reapply the strategy of Algorithm~\ref{alg:gather_specific}. However, for this algorithm to work correctly, it requires as input both the multiset of labels $L$ and the graph order $n$. In particular, termination detection relies on the knowledge of $L$: an agent stops when, at its current node, the multiset of labels of the advanced agents (those having at least completed the execution of the first line of the algorithm) matches $L$. True, such a simple detection could also have been done knowing only the cardinality of $L$. But in our situation, we are nowhere near knowing $L$ or its cardinality. When the oracle provides the binary representation of $\lceil\log\log (\mu+1)\rceil$, this does not even reveal the value of $\mu$: it only yields an approximate upper bound $U$ such that $\mu\leq U \leq (\mu+1)^2$.

Even setting aside the problem of initial knowledge, another difficulty remains: that of complexity. A quick analysis can show that the time required by Algorithm~\ref{alg:gather_specific} to gather all agents is far from being polynomial in $n$ and $|\lambda|$. Certainly, we could have improved the complexity to some extent, by designing the injective map $f$ more carefully, but this would have not been enough without more radical change. Indeed, our method of progressively constructing larger and larger agent groups until only one remains, by applying function ${\tt TZ}$, necessarily requires providing it with an input that is more than just the smallest label of the group, to avoid symmetries that could prevent gathering. Actually, the function must also take into account, in some way, the other labels in the group. This simple observation constitutes a fundamental barrier to reaching polynomial complexity in $n$ and $|\lambda|$ using this strategy. We therefore need a significantly more sophisticated approach, whose main ideas are outlined below. 

Our algorithm is structured around five states, between which an agent can switch and act accordingly. These states are: {\tt cruiser}, {\tt token}, {\tt explorer}, {\tt searcher}, and {\tt shadow}. Among them, {\tt shadow} and {\tt searcher} can be introduced immediately as they are simple and play only a secondary role in the overall process. When an agent enters state {\tt shadow}, it never leaves it: from then on, it only mirrors, in every round, the actions of a specific agent, called its \emph{guide}\textemdash exiting through the same port, waiting when the guide waits, and declaring termination when the guide does. When an agent is in state {\tt searcher}, it simply applies ${\tt EXPLO}(2^1)$, ${\tt EXPLO}(2^2)$, ${\tt EXPLO}(2^3)$, and so on, till finding an agent in state {\tt token}, in which case it transits to state {\tt shadow} by choosing this agent as its guide (this is unambiguous as we will prove that two agents can never be in state {\tt token} at the same node in the same round). By contrast, the states other than {\tt shadow} and {\tt searcher} play a central role and require a detailed explanation, as they capture the essence of the algorithm.

A fundamental requirement is that, within a polynomial number of rounds in $n$ and $|\lambda|$ from the first round $r_0$ in which any agent wakes up, every agent must have woken up. State {\tt cruiser} is a first step toward fulfilling this condition. Upon waking up, an agent begins in this state, where the primary objective is to ``quickly'' meet another agent. For simplicity, and because it is sufficient for our purpose, the agent in state {\tt cruiser} will precisely focus on finding a peer that is in state {\tt cruiser} or {\tt token}, and will deliberately ignore the meetings involving only agents in the other three states. To achieve this, the agent proceeds in phases $i= 1,2,\ldots$ in which it executes procedure ${\tt EXPLO}(2^i)$ and then, for $2^i$ rounds, procedure ${\tt TZ}$ using its own label as input. This continues until it finally meets an agent in one of the two above-mentioned states. What happens then depends on the states of the agents involved in the meeting. If the meeting involves an agent $X$ in state {\tt token}, then our agent in state {\tt cruiser} transits to state {\tt shadow} by choosing $X$ as its guide. Otherwise, all agents involved in the meeting are in state {\tt cruiser}: depending on their current memories (which are all different in view of Lemma~\ref{lem:diff}), one becomes {\tt explorer}, one becomes {\tt token}, while the others, if any, transit to state {\tt shadow} by choosing as their guide the agent becoming {\tt token}.

At this point, using the properties of procedures ${\tt EXPLO}$ and ${\tt TZ}$, together with the fact that an execution of the first $\lceil \log x\rceil$ phases lasts at most $H(x)$ rounds for some polynomial $H$, we can guarantee that the team's first meeting\textemdash necessarily involving only agents in state {\tt cruiser}\textemdash occurs within poly$(n,|\lambda|)$ time from $r_0$. Precisely, if some agent is not woken up by the adversary by round $r_0+H(\lceil \log n \rceil)$, then a meeting is guaranteed by this round because an execution of phase $\lceil\log n\rceil$ ensures that all nodes are visited via ${\tt EXPLO}(2^{\lceil\log n\rceil})$. Otherwise, all executions of {\tt TZ} initiated by the agents in a phase $i$ start within a window of fewer than $H(\lceil \log n \rceil)$ rounds, and, in view of the meeting bound of procedure  {\tt TZ} (cf. Section~\ref{sec:rela}), we can prove that a meeting occurs before any agent finishes the execution of {\tt TZ} in phase $\lceil \log n \rceil + \lceil (\alpha \log (n|\lambda|) \rceil$, for some constant $\alpha$ (i.e., by round $r_0+H(2^{\lceil \log n \rceil + \lceil (\alpha \log (n|\lambda|) \rceil})$. Hence, in either case, the team's first meeting indeed occurs in some round $r_1$, within poly$(n,|\lambda|)$ time from $r_0$. However, we cannot yet guarantee the fundamental requirement that, within such a time frame, every agent wakes up. To achieve this, additional mechanisms are needed. This is where the states {\tt token} and {\tt explorer} come into the picture.

When an agent $A$ becomes an explorer in some round $t$ at some node $v$, there is exactly one agent $B$ that becomes a token in the same round at the same node. As the name suggests, the role of $A$ in state {\tt explorer} is to explore the graph. To do so, it emulates procedure {\tt EST} (cf. Section~\ref{sec:rela}), using $B$ as its token. While in state {\tt token}, agent $B$ plays a passive role: it simply remains idle at node $v$. The only exception is a specific situation described later, in which it will interrupt this behavior and transition to state {\tt searcher}. Meanwhile, during its emulation of {\tt EST} started in round $t$, agent $A$ follows the instructions of {\tt EST}, but whenever it encounters an agent $C$ in state {\tt token} with $M_C(t) = M_B(t)$, it treats (rightly or wrongly) $C$ as its own token $B$ and proceeds accordingly. In the ideal situation where all agents have distinct labels, $A$ can never confuse $B$ with any other agent, and we can be certain that every node (resp. agent) is visited (resp. woken up) by the end of its exploration, which lasts at most $8n^5$ rounds\textemdash the worst-case complexity of {\tt EST}. In particular, the port-labeled BFS tree constructed by $A$ through this exploration indeed corresponds to a spanning tree of the underlying graph $G$. Unfortunately, we are not necessarily in such an ideal situation. While exploring the graph, agent $A$ might sometimes mistakenly consider another agent as its token $B$, which could prevent it from visiting all nodes and, consequently, we cannot guarantee that all agents are woken up by the end of its exploration. Likewise, the tree constructed by $A$ may actually correspond to a truncated BFS tree of the graph. However, what we can prove is that all explorers starting an exploration in round $t$ with a token whose memory matches that of $B$ in this round will, collectively, succeed in exploring the entire graph by round $t+8n^5$. Roughly speaking, the union of their explorations, and, more specifically, of the trees they will construct, will span the whole graph, even though some trees may cover significantly larger portions than others. As a result, the algorithm satisfies the aforementioned fundamental requirement: every agent is necessarily woken up within poly$(n,|\lambda|)$ time from $r_0$.

While this is of course a significant step forward, we need to keep in mind that our ultimate goal is much stronger and requires all agents to gather at the same node within poly$(n,|\lambda|)$ time (from $r_0$). Perhaps surprisingly, we are quite close to achieving this if, instead of limiting each explorer $A$ to a single exploration, they continue performing EST-based explorations until some conditions are met. Precisely, during each exploration that starts (resp. ends) in some round $t$ (resp. $t'$), agent $A$ considers (similarly as above) any agent $C$ in state {\tt token} as its own token $B$ if $M_C(t) = M_B(t)$. Furthermore, $A$ and its token $B$, which can be shown to be together in round $t'$, both switch to state {\tt searcher} in this round if, at some point in the interval $[t..t']$, $A$ has encountered an agent $C$ in state {\tt token} that entered this state before $B$, or in the same round as $B$ but $M_B(t)\prec M_C(t)$. With these simple rules and the fact the team's symmetry index is $1$, we can show the emergence of a kind of domino effect, by which, within poly$(n,|\lambda|)$ time, only a single explorer $E$ and a single token $T$ eventually remain, while all other agents become {\tt shadow} of this token. In particular, $E$ and $T$ can be shown to enter states {\tt explorer} and {\tt token}, respectively, in round $r_1$ at the same node, and remain in them forever thereafter. Since each exploration lasts at most $8n^5$ rounds, it follows that, within poly$(n,|\lambda|)$ time, there exists a round $r_2$ in which all agents are together at the same node, with $E$ in state {\tt explorer}, $T$ in state {\tt token}, and all other agents in state {\tt shadow}! However, a crucial problem persists: in the current design of our solution, the agents cannot detect this event, because they know a priori neither $n$ nor $\lambda$. To overcome this\textemdash and, notably, to prevent the last surviving explorer from performing explorations ad indefinitely\textemdash we relied on two additional fundamental findings that we brought to light in our analysis. In particular, this is where we make use of the upper bound $U$ on $\mu$, derived by an agent from the oracle's advice, which, as explained earlier, satisfies $\mu\leq U \leq (\mu+1)^2$.

The first finding is that when an explorer has completed $U^c$ consecutive identical explorations for some prescribed constant $c$, then it can get a polynomial estimation of $n$. (Two explorations are considered identical if, in both, the explorer follows exactly the same path and encounters an agent it interprets as its token at the exact same times.) More precisely, we can prove that when such an event occurs for an explorer $A$, the order $\eta$ of each of the identical (possibly truncated) BFS trees constructed during the $U^c$ consecutive identical explorations of $A$ satisfies $n\leq U\eta \leq (n+1)n^2$. Intuitively, this stems from the fact that, at some point, a group of at most $\mu$ explorers (including $A$) must have acted in a perfectly symmetric manner, notably by constructing identical trees (of order $\eta$) whose union spans the entire graph. Moreover, it is important to mention that starting from round $r_2$, the last surviving explorer can no longer confuse its token with another one and is therefore guaranteed to complete $U^c$ consecutive identical explorations by round $r_2+8U^cn^5$.

The second finding is that while we can prove that the difference $r_2-r_0$ is upper bounded by some polynomial $P(n,|\lambda|)$, we can also establish a more subtle and useful fact, which may seem counterintuitive at first: the difference $r_2-r_0$ is upper bounded by some polynomial $P'(n,\tau)$, where $\tau$ is the minimum time spent by an agent in state {\tt cruiser}. Why it is more subtle and useful? Consider any explorer $X$ that completes $U^c$ consecutive identical explorations. Let $\eta_X$ be the order of the trees computed in these explorations and $\tau_X$ the time spent by $X$ in state {\tt cruiser}. Since $n\leq U\eta_X\leq (n+1)n^2$ and $\tau\leq \tau_X\leq P(n,|\lambda|)$, we necessarily have: $$P'(n,\tau)\leq P'(U\eta_X,\tau_X)\leq P'((n+1)n^2,P(n,|\lambda|))$$

In other words, once an explorer $X$ obtains its estimated upper bound $U\eta_X$ on the graph order, it can combine it with $\tau_X$ to determine an upper bound on the duration that must elapse before exceeding round $r_2$. And more importantly, all such upper bounds on this duration computed by the explorers are themselves upper bounded by the same polynomial in $n$ and $|\lambda|$, without requiring prior knowledge of $n$ or $\lambda$!

Consequently, we can ensure that gathering with detection is achieved within poly$(n,|\lambda|)$ time, by asking each explorer $X$ to declare termination at the end of an exploration if:
\begin{enumerate}[label=(\roman*)]
\item it has just completed at least $U^c$ consecutive identical explorations, each producing a tree of order $\eta_X$, and 
\item the total number of elapsed rounds since the beginning of its 
execution exceeds some polynomial bound in both the estimated graph 
size $U\eta_X$ and the time $\tau_X$ it spent in state {\tt cruiser}.\footnote{For technical reasons, we substituted this second condition in our algorithm with carefully parametrized waiting periods inserted between explorations.}
\end{enumerate}

In practice, only the last surviving explorer eventually satisfies these conditions, at which point it is co-located with its token and all other agents (in state {\tt shadow}). Of course, when the explorer declares termination, the algorithm will ask all the colocated agents to simultaneously do the same.

Overall, we end up with an algorithm\textemdash called ${\mathcal{HG}}$\textemdash whose design is remarkably simple. However, this simplicity comes at a price of a difficult and challenging analysis. One of the main difficulties was undoubtedly to show that our termination detection mechanism, relying on the two findings described above, indeed works within the expected time frame without producing any false positives.

\subsection{Algorithm}

In this section, we give a detailed description of procedure ${\mathcal{HG}}$ executed by an agent $A$ of label $\ell_A$. For ease of reading, we describe it as a list of states in which agent $A$ can be and, for each of them, we specify the actions performed by $A$ as well as the rules governing the transitions to other states. In the description of each state $X$, when we say ``agent $A$ transits to state $Y$'' in a round $r$, we exactly mean that agent $A$ remains idle in its current node in state $X$ until the end of round $r$ and enters state $Y$ in round $r+1$ (precisely at the beginning of $r+1$). By doing so, each transition costs one extra round, but we have the guarantee that, in each round, agent $A$ is in at most one state. The states in which agent $A$ can be are: {\tt cruiser}, {\tt token}, {\tt explorer}, {\tt searcher} and {\tt shadow}. Upon waking up, the agent starts in state {\tt cruiser}. (If an agent is dormant or has completed the execution of ${\mathcal{HG}}$, it is considered to have no state.)\\

  \noindent
 {\bf State} {\tt cruiser}.

In this state, agent $A$ works in phases $i= 1,2,\ldots$ In phase $i$, it executes procedure ${\tt EXPLO}(2^i)$ and then, for $2^i$ rounds, procedure ${\tt TZ}(\ell_A)$. As soon as it meets another agent in state {\tt cruiser} or {\tt token} in a round $r$ at a node $v$, it interrupts the execution of the current phase and transits either to state {\tt shadow} or {\tt explorer} or {\tt token} in the same round, depending on the cases presented below. Let $\mathcal{S}$ be the set of all agents in state {\tt cruiser} or {\tt token} located at node $v$ in round $r$, including agent $A$.

\begin{itemize}
\item{\it Case~1. All agents of $\mathcal{S}$ are in state {\tt cruiser} in round $r$.} In this case, there are at least two agents of $\mathcal{S}$ that are in state {\tt cruiser} in round $r$. If agent $A$ has the largest (resp. the smallest) memory within $\mathcal{S}$ in round $r$, then it transits to state {\tt token} (resp. {\tt explorer}). Otherwise, it transits to state {\tt shadow} and its guide is the agent having the largest memory within $\mathcal{S}$ in round $r$ (the role of the guide is given in the description of state {\tt shadow}).

\item{\it Case~2. There is an agent of $\mathcal{S}$ that is in state {\tt token} in round $r$.} Agent $A$ transits to state {\tt shadow} and its guide is the agent of $S$ that is in state {\tt token} in round $r$ (we will show in Section~\ref{sec:correct} that, in every round, a node can contain at most one agent in state {\tt token}).

\end{itemize} 

  \noindent
 {\bf State} {\tt searcher}.

In state {\tt searcher}, agent $A$ executes successively ${\tt EXPLO}(2)$, ${\tt EXPLO}(2^2)$, ${\tt EXPLO}(2^3)$, \ldots, ${\tt EXPLO}(2^i)$, and so on, until it meets an agent in state {\tt token}. When such a meeting occurs at a node $v$ in a round $r$, agent $A$ transits to state {\tt shadow} in the same round and selects as its guide the agent that is in state {\tt token} at node $v$ in round $r$.\\

  \noindent
 {\bf State} {\tt shadow}.

We will show that when agent $A$ enters state {\tt shadow} at a node $v$, it has exactly one guide that is also located at node $v$. At the end of each round, agent $A$ simply takes the same decision as its guide regarding whether to remain idle, take a specific port, or declare termination (we will see that agent $A$ always stays colocated with its guide and can always identify it unambiguously). If, in a round $r$, $A$'s guide decides to transit itself to state {\tt shadow} and selects an agent $C$ as its guide, then, in round $r+1$, $A$'s guide is $C$.\\


  \noindent
{\bf State} {\tt token}

While in this state, agent $A$ remains idle and does nothing until it is, in a round $r$, in one of the following two situations:
\begin{enumerate}
\item agent $A$ shares its current node with an agent in state {\tt explorer} that transits to state {\tt searcher},
\item or all agents in state {\tt explorer} at the node occupied by $A$ declare termination.\footnote{We write ``all agents in state {\tt explorer}'' instead of ``an agent in state {\tt explorer}'' to formally avoid what might appear as a conflict in the description\textemdash namely, at the occupied by $A$, one {\tt explorer} switching to state {\tt searcher} while another declares termination. In practice, however, this can never happen: we prove that only one agent in state {\tt explorer} will declare termination, and at that moment no other agents are in this state.}
\end{enumerate}

In the first situation (resp. second situation) agent $A$ also transits to state {\tt searcher} (resp. declares termination) in round $r$.\\

  \noindent
{\bf State} {\tt explorer}

We will show that when agent $A$ enters state {\tt explorer} at a node $v$, there is exactly one agent $B$ that simultaneously enters state {\tt token} at the same node. We will also show that $B$ remains in state {\tt token} (and thus idle at node $v$) as long as $A$ remains in state {\tt explorer}. In this state, agent $A$ executes the protocol given in algorithm~\ref{alg:algexplo}, which uses the notion of \emph{seniority} that is defined as follows. The seniority of an agent $B$ in state {\tt token} (resp. {\tt explorer}) in a round $r'$ is the difference $r'-r$, where $r$ is the round in which $B$ entered state {\tt token} (resp. state {\tt explorer}). Since, in view of the description of the states, an agent can enter state {\tt token} or {\tt explorer} at most once, the difference is well defined. The protocol also uses function ${\tt EST}^+$, which corresponds to procedure {\tt EST} but with the following three changes. Consider an execution of ${\tt EST}^+$ initiated by agent $A$ in a round $t$ from a node $u$. The first change is that to start this execution, we need to give function ${\tt EST}^+$ any memory $M$ as input. The second change is that each time agent $A$ encounters an agent $B$ that is in state {\tt token} and such that $M_B(t)=M$,\footnote{Note that even if agent $A$ is not with $C$ in round $t$, it can determine $M_C(t)$ when meeting $C$ in a round $t'>t$ using $M_C(t')$} agent $A$ considers it is with its token. The third change is that ${\tt EST}^+$ returns a triple $(b,\eta,trace)$. The first element, $b$, is a Boolean. Its value is true if, during the execution of ${\tt EST}^+$, agent $A$ encounters an agent $C$ in state {\tt token} of either higher seniority or equal seniority but such that $M \prec M_C(t)$. Otherwise, $b$ is false. The second element, $\eta$, corresponds to the number of nodes in the BFS tree constructed during the execution of ${\tt EST}^+$ (recall that, when executing ${\tt EST}$, an agent constructs such a tree). Finally the third element, $trace$, is, as the name suggests, a trace of the execution. It is represented as the sequence $(p_1,q_1,m_1,p_2,q_2,m_2,\ldots,p_k,q_k,m_k)$ that satisfies the following two conditions. The first condition is that $(p_1,q_1,p_2,q_2,\ldots,p_k,q_k)$ is the path followed by agent $A$ during the execution of ${\tt EST}^+$. The second condition is that for every $1\leq i\leq k$, $m_i=1$ (resp. $m_i=0$) if in round $t+i$ agent $A$ is (resp. is not) with a token (in the sense of the second change given above).

\begin{algorithm}[H]
\label{alg:algexplo}
\caption{Procedure executed by agent $A$ while in state {\tt explorer}}
\SetNoFillComment
$counter :=1$\;
$x := $ the decimal number whose binary representation is $\mathcal{K}$\label{l:2}\; 
$U := 2^{(2^x)}$; $trace_{old} :=$ the empty list\label{l:3}\;
$\tau :=$ the number of rounds spent by agent $A$ in state {\tt cruiser}\; 
\Repeat{$counter= U^{25\beta}$ \OR $b = \true$\label{l:repeat}}{
$M :=$ the memory of the agent in state {\tt token} that is currently at the same node as $A$\label{l:4}\;
$(b,\eta,trace_{new}) :=$ {\tt EST$^+(M)$}\label{l:5}\;
\tcc{Agent $A$ is back at the node from which it has started the previous execution of EST$^+$}
\If{$b = \false$\label{l:if}}%
{
wait $\sum_{i=1}^{11\lceil \beta \log (\eta U\tau) \rceil}$ $2^i$ rounds\label{l:if2} \;
\eIf{$trace_{old}=trace_{new}$}
{$counter := counter+1$\;}
{$counter := 1$\;}
$trace_{old} := trace_{new}$\;
}
}
\eIf{$b = \true$\label{l:afterrepeat1}}
{transit to state {\tt searcher}\;\label{l:afterrepeat2}}
{declare termination\;\label{l:afterrepeat3}}
\end{algorithm}

~\\

\subsection{Correctness and Complexity Analysis}
\label{sec:correct}
All the lemmas given in this subsection are stated under the following three assumptions. The first assumption is that the team of agents is obviously gatherable. The second assumption is that the algorithm executed by an agent when it wakes up is procedure ${\mathcal{HG}}$. The third assumption is that the advice $\mathcal{K}$, given by the oracle, corresponds to the binary representation of $\lceil\log\log (\mu+1)\rceil$, where $\mu$ is the multiplicity index of the team. Thus, for ease of reading, these assumptions will not be explicitly repeated in the statements of the lemmas.

Furthermore, since $\mathcal{K}$ corresponds to the binary representation of $\lceil\log\log(\mu+1)\rceil$, it follows that variable $U$ used in Algorithm~\ref{alg:algexplo} is necessarily equal to $2^{(2^{\lceil\log\log(\mu+1)\rceil})}$. Thus, for simplicity, we will reuse $U$ as the notation for this value throughout the proof below. Finally, we will denote by $r_0$ the first round in which at least one agent wakes up. 

We start with the following lemma that is a direct consequence of procedure $\mathcal{HG}$.

\begin{lemma}
\label{lem:onlyone}
Let $S$ be any of the $5$ states of procedure $\mathcal{HG}$. If an agent leaves state $S$ in a round $r$, it cannot return to this state in any subsequent round. 
\end{lemma}

The next two lemmas establish key properties related to state {\tt token}.

\begin{lemma}
\label{lem:tokappears}
There exists an agent that enters state {\tt token} within at most a polynomial number of rounds in $n$ and $|\lambda|$ after round $r_0$.
\end{lemma}

\begin{proof}
In Section~\ref{sec:preli}, we mention that if two agents $X$ and $X'$ start executing procedure ${\tt TZ}$ in possibly different rounds, then they will meet after at most $(n\cdot\min(|\ell_X|,|\ell_{X'}|))^{\alpha}$ rounds since the start of the later agent, for some integer constant $\alpha\geq 2$. We also mention that for any positive integer $N\geq2$, procedure ${\tt EXPLO}(N)$ lasts at most $N^{\beta}$ rounds, for some integer constant $\beta\geq 2$.

In this proof, constants $\alpha$ and $\beta$  play a central role. Indeed, we show below that some agent necessarily enters state {\tt token} by round $r_0+(16T)^{\beta}$ with $T=(n|\lambda|)^{\alpha}+(8n)^{\beta}$, which is enough to prove the lemma.

Assume by contradiction that no agent enters state {\tt token} by round $r_0+(16T)^{\beta}$. We have the following claim.

\begin{claim}
\label{claim1}
Let $X$ be any agent that wakes up in a round $r_0 \leq r \leq r_0+(16T)^{\beta}$. Agent $X$ never meets an agent in state {\tt token} or {\tt cruiser} before round $r_0+(16T)^{\beta}$ and it remains in state {\tt cruiser} from round $r$ to round $r_0+(16T)^{\beta}$ included.
\end{claim}

\begin{proofclaim}
Suppose by contradiction that the claim does not hold. According to Algorithm $\mathcal{HG}$, an agent starts in state {\tt cruiser} in the round of its wake-up and decides to transit from state {\tt cruiser} to another state in some round if, and only if, it meets an agent in state {\tt token} or {\tt cruiser} in this round. Moreover, recall that when an agent decides to transit to a new state in some round, it enters this state in the next round. Hence, there is a round $r\leq r' \leq r_0+(16T)^{\beta}-1$ in which agent $X$ decides to transit from state {\tt cruiser} to another state because, in round $r'$, it is at some node $v$ with an agent in state {\tt token} or {\tt cruiser}. In the first case, it means that an agent enters state {\tt token} by round $r' \leq r_0+(16T)^{\beta}-1$, which contradicts the assumption that no agent enters state {\tt token} by round $r_0+(16T)^{\beta}$. Hence, at node $v$ in round $r'$, there is no agent in state {\tt token} and there are at least two agents in state {\tt cruiser} (including agent $X$). Let $\mathcal{S}$ be the set of all agents in state {\tt cruiser} at node $v$ in round $r$. In view of Lemma~\ref{lem:diff}, no two agents of $\mathcal{S}$ can have the same memory in round $r$. Hence, according to the rules applied by an agent in state {\tt cruiser}, the agent having the largest memory in round $r$ among the agents of $\mathcal{S}$ necessarily enters state {\tt token} in round $r'+1\leq r_0+(16T)^{\beta}$. This contradicts again the assumption that no agent enters state {\tt token} by round $r_0+(16T)^{\beta}$, which concludes the proof of this claim.
\end{proofclaim}

~\newline

While in state {\tt cruiser}, any agent $X$ works in phases $i= 1,2,\ldots$ where phase $i$ consists of an execution of procedure ${\tt EXPLO}(2^i)$ followed by an execution, for $2^i$ rounds, of procedure ${\tt TZ}(\ell_X)$. All of this is interrupted only when it meets another agent in state {\tt cruiser} or {\tt token}. For every positive integer $k$ an entire execution of the first $\lceil\log k\rceil$ phases of state {\tt cruiser} lasts at most $\sum_{i=1}^{\lceil\log k\rceil}(2^{i\beta}+2^i)$, which is upper bounded by $(8k)^{\beta}$. Hence, an entire execution of the first $\lceil\log n\rceil$ phases lasts at most $(8n)^{\beta}$. Let $A$ be any agent that wakes up in round $r_0$. 
Since $(8n)^{\beta}<(16T)^{\beta}$, we know by Claim~\ref{claim1} that agent $A$ is in state {\tt cruiser} from round $r_0$ to round $r_0+(8n)^{\beta}$ included. In view of the explanations given just above, it follows that an entire execution of phase $\lceil\log n \rceil$, and thus of procedure ${\tt EXPLO}(2^{\lceil\log n \rceil})$, is completed by agent $A$ by round $r_0+ (8n)^{\beta}$. Since an execution of procedure ${\tt EXPLO}(2^{\lceil\log n \rceil})$ allows to visit every node of the underlying graph at least once, it follows that every agent enters state {\tt cruiser} by round $r_0+(8n)^{\beta}$. Indeed, otherwise it means that by some round $r\leq r_0+ (8n)^{\beta}$, agent $A$ enters a node occupied by a dormant agent that immediately enters state {\tt cruiser} in round $r$, which contradicts Claim~\ref{claim1}. As a result, we have the guarantee that the delay between the wake-ups of any two agents is at most $(8n)^{\beta}$ rounds and, in view of Claim~\ref{claim1}, every agent is in state {\tt cruiser} from round $r_0 + (8n)^{\beta}$ to round $r_0+(16T)^{\beta}$ included. 

We stated above that, for every positive integer $k$, an entire execution of the first $\lceil\log k\rceil$ phases of state {\tt cruiser} is upper bounded by $(8k)^{\beta}$. This implies that an entire execution of the first $\lceil\log T\rceil$ phases is upper bounded by $(8T)^{\beta}$. Moreover, $(8n)^{\beta}+(8T)^{\beta}<(16T)^{\beta}$, which means that every agent is in state {\tt cruiser} from round $r_0 + (8n)^{\beta}$ to round $r_0+(8n)^{\beta}+(8T)^{\beta}$ included. Hence, an entire execution of phase $\lceil\log T \rceil$ is completed by every agent by round  $r_0+(8n)^{\beta}+(8T)^{\beta}$. During its execution of phase $\lceil\log T \rceil$, any agent $X$ starts an execution of procedure ${\tt TZ}(\ell_X)$ for at least $T\geq (n|\lambda|)^{\alpha}+(8n)^{\beta}$ rounds. Since the delay between the wake-ups of any two agents is at most $(8n)^{\beta}$, the delay between the starting times of procedure ${\tt TZ}$ by any two agents in phase $\lceil\log T \rceil$ is also at most $(8n)^{\beta}$ rounds. Furthermore, by Theorem~\ref{theo:charac} and the assumption that the team is gatherable, it follows that the team must contain at least two agents with distinct labels. Hence, there is an agent $B$ with label $\ell_B=\lambda$ (resp. an agent $B'$ with a label $\ell_{B'}\ne\lambda$) that executes  ${\tt TZ}(\lambda)$ (resp. ${\tt TZ}(\ell_{B'})$), from round $r'$ to round $r'+ (n|\lambda|)^{\alpha}$ included, where $r'$ is the later of the starting rounds of the two agents of procedure ${\tt TZ}$ in phase $\lceil\log T \rceil$. According to the properties of procedure {\tt TZ} recalled at the beginning of this proof, agent $B$ meets agent $B'$ in some round $r'\leq r''\leq r'+ (n|\lambda|)^{\alpha}$. Note that agent $B$ and $B'$ are necessarily woken up in round $r''$. Besides, $r'+ (n|\lambda|)^{\alpha}<r_0+(8n)^{\beta}+(8T)^{\beta}$ as an entire execution of phase $\lceil\log T \rceil$ is completed by every agent by round  $r_0+(8n)^{\beta}+(8T)^{\beta}$. In view of this and Claim~\ref{claim1},  agent $B$ and $B'$ are both in state {\tt cruiser} when they meet in round $r''$ and $r''< r_0+(16T)^{\beta}$. This is a contradiction with Claim~\ref{claim1}. Hence, some agent enters state {\tt token} by round $r_0+(16T)^{\beta}$, which concludes the proof of the lemma.
\end{proof}

\begin{lemma}
\label{lem:notwotok}
An agent enters state {\tt token} (resp. {\tt explorer}) at a node $v$ in a round $r$ if, and only if, exactly one agent enters state {\tt explorer} (resp. {\tt token}) at node $v$ in round $r$. Moreover, in every round, there can never be two agents in state {\tt token} that occupy the same node.
\end{lemma}

\begin{proof}
According to Algorithm~$\mathcal{HG}$, an agent $A$ can enter state {\tt token} (resp. {\tt explorer}) at a node $v$ in a round $r$ if, and only if, the following three conditions are satisfied at node $v$ in round $r-1$: 
\begin{enumerate}
\item The set $\mathcal{S}$ of the agents in state {\tt cruiser} contains at least two agents, including agent $A$.
\item Agent $A$ has the largest memory (resp. the smallest memory) among the agents of $\mathcal{S}$. 
\item There is no agent in state {\tt token}. 
\end{enumerate}

Moreover, Lemma~\ref{lem:diff} implies that the memories of any two agents of $\mathcal{S}$ in round $r-1$ are different. Hence, an agent enters state {\tt token} (resp. {\tt explorer}) at a node $v$ in a round $r$ if, and only if, exactly one agent enters state {\tt explorer} (resp. {\tt token}) at node $v$ in round $r$.

To prove the lemma, it now remains to show that in every round, there can never be two agents in state {\tt token} that occupy the same node. Assume, by contradiction, there is a round $r'$ in which two agents $A$ and $B$ are in state {\tt token} at the same node $u$. Let $r'_A$ (resp. $r'_B$) be the round in which agent $A$ (resp. $B$) enters state {\tt token} ($r'_A$, as well as $r'_B$, is  unique and at most equal to $r'$ in view of Lemma~\ref{lem:onlyone}). According to what is stated above and the fact that an agent never moves while in state {\tt token}, we know that agent $A$ (resp. $B$) is in state {\tt cruiser} at node $u$ in round $r'_A-1$ (resp. $r'_B-1$) and has the largest memory among the agents in state {\tt cruiser} at node $u$ in round $r'_A-1$ (resp. $r'_B-1$). We also know that agent $A$ (resp. $B$) is in state {\tt token} at node $u$ from round $r'_A-1$ (resp. $r'_B-1$) to round $r'$ included.

If $r'_A=r'_B$ then agents $A$ and $B$ are both in state {\tt cruiser} at node $u$ in round $r'_A-1$ and both have the largest memory among the agents in state {\tt cruiser} (and thus the same memory) at node $u$ in round $r'_A-1$. However, this directly contradicts Lemma~\ref{lem:diff}. 

Consequently, we necessarily have $r'_A\ne r'_B$ and we can suppose, without loss of generality, that $r'_A<r'_B$. It follows that at node $u$ in round $r'_B-1$, agent $A$ (resp. $B$) is in state {\tt token} (resp. {\tt cruiser}). However, according to the three conditions given at the beginning of this proof, agent $B$ can enter state {\tt token} in round $r'_B$ at node $u$ only if there is no agent in state {\tt token} in round $r'_B-1$ at node $u$. This means that agent $B$ cannot enter state {\tt token} in round $r'_B$: this is a contradiction with the definition of round $r'_B$.

As a result, in every round, there can never be two agents in state {\tt token} that occupy the same node, which ends the proof of the lemma.
\end{proof}

Consider any agent $A$ (resp. $B$) that ends up entering state {\tt explorer} (resp. {\tt token}). By Lemma~\ref{lem:onlyone}, we know that agent $A$ (resp. $B$) enters state {\tt explorer} (resp. {\tt token}) exactly once. By Algorithm~\ref{alg:algexplo}, we also know that once an agent becomes {\tt explorer} (resp. {\tt token}), it will never enter state {\tt token} (resp. {\tt explorer}) thereafter. Hence, we can unambiguously define ${\tt home}(A)$ (resp. ${\tt home}(B)$) as the node at which agent $A$ (resp. $B$) enters state {\tt explorer} (resp. {\tt token}). Moreover, when it does so, there is exactly one agent that simultaneously enters state {\tt token} (resp. {\tt explorer})  at ${\tt home}(A)$ (resp. ${\tt home}(B)$) in view of Lemma~\ref{lem:notwotok}. In the remainder, this other agent will be denoted by ${\tt tok}(A)$ (resp. ${\tt exp}(B)$) (since agent $A$ (resp. $B$) enters state {\tt explorer} (resp. {\tt token}) exactly once, the definition of ${\tt tok}(A)$ (resp. ${\tt exp}(B)$) is unambiguous). By a slight misuse of language, we will occasionally refer to ${\tt tok}(A)$ as the token of $A$ and ${\tt exp}(B)$ as the explorer of $B$. Note that we have ${\tt home}(A)={\tt home}(B)$ if $B={\tt tok}(A)$ (or equivalently if $A={\tt exp}(B)$).
Having established these definitions, we can now proceed to the next lemma.

\begin{lemma}
\label{lem:exp}
Let $\mathcal{E}$ be an execution of ${\tt EST}^+$ started in a round $r$ by an agent $A$. We have the following four properties.
\begin{enumerate}
\item In $\mathcal{E}$, the input given to function ${\tt EST}^+$ is $M_{{\tt tok}(A)}(r)$.
\item $\mathcal{E}$ lasts at most $8n^5$ rounds.
\item In round $r$ and in the round when $\mathcal{E}$ is completed, agent $A$ (resp. ${\tt tok}(A)$) is in state {\tt explorer} (resp. {\tt token}) at ${\tt home}(A)$.
\item Let $(*,\eta,*)$ be the triple returned by function ${\tt EST}^+$ when $\mathcal{E}$ is completed. If there exists no agent $X\ne {\tt tok}(A)$ such that $M_X(r)=M_{{\tt tok}(A)}(r)$, then all nodes of the underlying graph $G$ are visited by $A$ during $\mathcal{E}$ and $\eta=n$, otherwise $1\leq\eta\leq n$.
\end{enumerate}
\end{lemma}

\begin{proof}
Before starting, just observe that an agent can start an execution of ${\tt EST}^+$ only when in state {\tt explorer} and it cannot decide to leave this state during an execution of ${\tt EST}^+$. This implies that agent $A$ is in state {\tt explorer} in round $r$ and in the round when $\mathcal{E}$ is completed if such a round exists (of course, we show below that such a round exists). This observation must be kept in mind when reading the proof as it will not always be repeated in order to lighten the text. 

We now start this proof by assuming that property $\Psi(\mathcal{E},A)$, made of the following two conditions, is satisfied:

\begin{itemize}
\item In round $r$, agent $A$ is at ${\tt home}(A)$  with ${\tt tok}(A)$ that is in state {\tt token} and
\item during $\mathcal{E}$, agent ${\tt tok}(A)$ never decides to switch from state {\tt token} to state {\tt searcher} or to declare termination.
\end{itemize} 

Our goal is to first establish that, under this assumption, the lemma holds, and then to demonstrate that property $\Psi(\mathcal{E},A)$ is necessarily satisfied. Note that property $\Psi(\mathcal{E},A)$ immediately implies that ${\tt tok}(A)$ always remains in state {\tt token} at ${\tt home}(A)$ during $\mathcal{E}$ and, in particular, in the round when $\mathcal{E}$ is completed, if any. Also note that property $\Psi(\mathcal{E},A)$ and Lemma~\ref{lem:notwotok} ensure that ${\tt tok}(A)$ is the only agent in state {\tt token} at ${\tt home}(A)$ in round $r$. Thus, by lines~\ref{l:4} and~\ref{l:5} of Algorithm~\ref{alg:algexplo}, the input given to function ${\tt EST}^+$ in $\mathcal{E}$ is necessarily well defined and indeed corresponds to $M_{{\tt tok}(A)}(r)$, which proves the first property of the lemma.

As mentioned in Section~\ref{sec:preli}, an execution of procedure ${\tt EST}$ by an agent $Y$ in a setting where a single token is present in the underlying graph $G$, at the node from which $Y$ starts this procedure, allows this agent to visit all nodes of $G$ in at most $8n^5$ rounds. Additionally, at the end of the execution, agent $Y$ is back with its token and has constructed a BFS spanning tree of $G$. Hence, if there is no agent $X\ne {\tt tok}(A)$ such that $M_X(r)=M_{{\tt tok}(A)}(r)$, agent $A$ can never confuse its token with another one during $\mathcal{E}$ and thus, in view of property $\Psi(\mathcal{E},A)$, the second, third and fourth properties of the lemma also hold.

Consequently, suppose that there is at least one agent $X\ne {\tt tok}(A)$ such that $M_X(r)=M_{{\tt tok}(A)}(r)$. In this case, agent $A$ may sometimes mistakenly confuse its token ${\tt tok}(A)$ with another agent. By taking a close look at the description of procedure ${\tt EST}$ (cf. Section~\ref{sec:preli}), on which ${\tt EST}^+$ is based, we can observe that, during $\mathcal{E}$, agent $A$ never takes into account the presence or absence of ${\tt tok}(A)$ except in one specific situation. This occurs only when it decides, during the process of a node $c(x)$ of the BFS tree $T$ under construction, corresponding to a node $x$ of $G$, whether it adds to $T$ a node $c(x')$ corresponding to a neighbor $x'$ of $x$ (i.e, whether a node corresponding to $x'$ has already been added before or not). This means that confusing its token ${\tt tok}(A)$ with another agent may only occur at times of these decisions. However, while such a confusion can lead agent $A$, during the process of $c(x)$, to wrongly reject $x'$ if no corresponding node has not yet been added to the tree, it cannot lead the agent to wrongly admit $x'$ i.e., adding to $T$ a node corresponding to $x'$ while such a node has been previously added. Nor can such a confusion prevent the process of $c(x)$ from being completed at the node from which it is started, i.e., $x$. Also observe that execution $\mathcal{E}$ consists of alternating periods of two different types. The first type corresponds to periods when agent $A$ processes a node of the BFS tree. The second type corresponds to periods when:

\begin{itemize}
\item Agent $A$, having just finished processing a node of $T$ corresponding to some node $w$ in $G$, moves from $w$ to a node corresponding to an unprocessed node of $T$, following a path from this tree.
\item Or there is no remaining node to process and agent $A$, having just finished the last process of a node of $T$ corresponding to some node $w$ in $G$, moves from $w$ to the node corresponding to the root of $T$.
\end{itemize}

From these explanations, property $\Psi(\mathcal{E},A)$ and the fact that the root of the BFS tree $T$ constructed by $A$ corresponds to the node from which $\mathcal{E}$ is started by $A$ (actually it is the first node to be processed), we get two consequences. The first consequence is that the tree constructed by agent $A$ during $\mathcal{E}$ is either a spanning tree of $G$ or a non-empty truncated spanning tree of $G$ (i.e., a non-empty tree that can be obtained from a spanning tree of $G$ by removing one or more of its subtrees). The second consequence is that execution $\mathcal{E}$ is eventually completed at the node from which $\mathcal{E}$ is started by $A$. In view of these two consequences and property $\Psi(\mathcal{E},A)$, the third and fourth properties of the lemma hold. Concerning the second property, note that, whether or not there is at least one agent $X\ne {\tt tok}(A)$ such that $M_X(r)=M_{{\tt tok}(A)}(r)$, a period of the first (resp. second) type, during $\mathcal{E}$, takes at most $4n^3$ (resp. $2n$) rounds. Moreover, in $\mathcal{E}$, the number of periods of the first type, as well as the number of periods of the second type, is equal to the number of nodes that are added to BFS tree, i.e., at most $n$ (we explained above that the tree spans all or part of $G$). Thus, $\mathcal{E}$ indeed lasts at most $8n^5$ rounds, which proves the second property. This concludes the proof that the lemma is true if property $\Psi(\mathcal{E},A)$ is satisfied.

It now remains to show that property $\Psi(\mathcal{E},A)$ is necessarily satisfied. We will say that a round $r^*$ invalidates $\Psi(\mathcal{E},A)$ if at least one of the following conditions is met:

\begin{itemize}
\item $\mathcal{E}$ starts in round $r^*$ and either agent $A$ is not at ${\tt home}(A)$ in round $r^*$ or ${\tt tok}(A)$ is not in state {\tt token} at ${\tt home}(A)$ in round $r^*$.
\item $\mathcal{E}$ starts no later than round $r^*$ but is not completed by round $r^*$, and, in this round, ${\tt tok}(A)$ either decides to switch from state {\tt token} to state {\tt searcher} or to declare termination.
\end{itemize}  

Assume by contradiction that $\Psi(\mathcal{E},A)$ is not satisfied. This implies that there exists a round $r^*$ that invalidates $\Psi(\mathcal{E},A)$. Also assume, without loss of generality, that $r^*$ is minimal in the following precise sense: for every execution $\mathcal{E}'$ of function ${\tt EST}^+$ by any agent $A'$ in state {\tt explorer}, if a round $r'$ invalidates $\Psi(\mathcal{E}',A')$ then $r^*\leq r'$. (Note that $\mathcal{E}'$ (resp. $A'$) is not necessarily different from $\mathcal{E}$ (resp. $A$).) 

First consider the case in which $\mathcal{E}$ starts in round $r^*$
and either agent $A$ is not at ${\tt home}(A)$ in round $r^*$ or ${\tt
  tok}(A)$ is not in state {\tt token} at ${\tt home}(A)$ in round
$r^*$. Note that the first execution of function ${\tt EST}^+$ by
agent $A$ starts in the round when agent $A$ enters state {\tt
  explorer} at ${\tt home}(A)$. Moreover, by the definition of ${\tt
  tok}(A)$, ${\tt tok}(A)$ enters state {\tt token} at ${\tt home}(A)$
in the round when $A$ enters state {\tt explorer}. Hence,
$\mathcal{E}$ cannot correspond to the first execution of function
${\tt EST}^+$ by agent $A$. In other words, $\mathcal{E}$ corresponds
to the $i$th execution of function ${\tt EST}^+$ by agent $A$, for
some $i\geq 2$.

Denote by $\mathcal{E}'$ the $(i-1)$th execution of function ${\tt EST}^+$ by agent $A$ and denote by $r'_1$ (resp. $r'_2$) the round in which $\mathcal{E}'$ starts (resp. is completed). In view of Algorithm~$\mathcal{HG}$, agent $A$ can execute function ${\tt EST}^+$ only when in state {\tt explorer}, which, by Lemma~\ref{lem:onlyone}, it enters at most once. In particular, this means that agent $A$ is in state {\tt explorer} from round $r'_1$ to (at least) round $r^*$ included. By the definition of $\mathcal{E}'$, we have $r'_1\leq r'_2\leq r^*$, and thus, by the minimality of $r^*$, we get the guarantee that no round invalidates property $\Psi(\mathcal{E}',A)$. From this and the fact, proven above, that the four properties of the lemma hold if $\Psi(\mathcal{E}',A)$ is satisfied, we know that the triple $(b,\eta,*)$ returned by function ${\tt EST}^+$, when $\mathcal{E}'$ is completed in round $r'_2$, is such that $1\leq\eta\leq n$. By line~\ref{l:repeat} of Algorithm~\ref{alg:algexplo} and the fact that agent $A$ is in state {\tt explorer} from round $r'_1$ to (at least) round $r^*$ included, we also know that the first element $b$ of the triple is the Boolean value false. Hence, by lines~\ref{l:if} and \ref{l:if2} of Algorithm~\ref{alg:algexplo}, we know that agent $A$ starts in round $r'_2$ a waiting period of some positive duration (the duration is well defined and positive, since $1\leq\eta\leq n$ and $U=2^{(2^{\lceil\log\log (\mu+1)\rceil})}\geq 2$). This waiting period is necessarily completed in round $r^*$ because agent $A$ is in state {\tt explorer} from round $r'_2$ to (at least) round $r^*$ included and the next execution of function ${\tt EST}^+$ by agent $A$, after $\mathcal{E}'$, is $\mathcal{E}$ that starts in round $r^*$. This implies that agent $A$ is at ${\tt home}(A)$ in round $r^*$, as agent $A$ is already at ${\tt home}(A)$ in round $r'_2$ by the fact that no round invalidates property $\Psi(\mathcal{E}',A)$.

Consequently, for the case under analysis, we know that ${\tt tok}(A)$ is not in state {\tt token} at ${\tt home}(A)$ in round $r^*$. This means that, in some round $r''<r^*$ at node ${\tt home}(A)$, ${\tt tok}(A)$ decides to switch from state {\tt token} to state {\tt searcher} or to declare termination. According to the description of states {\tt token} and {\tt explorer}, this is due to the fact that an agent $A''$ in state explorer has completed an execution $\mathcal{E}''$ of function ${\tt EST}^+$ at ${\tt home}(A)$ in some round  $r'''\leq r''$ and then has waited $r''-r'''$ rounds (while staying in state {\tt explorer}) before deciding in round $r''$ to switch from state {\tt explorer} to state {\tt searcher} or to declare termination. Note that, in view of the minimality of $r^*$, $\Psi(\mathcal{E}'',A'')$ necessarily holds. Therefore, according to what has been proved in the first part of the proof, we can state that, in round $r'''$, agent $A''$ is with ${\tt tok}(A'')$ at ${\tt home}(A'')$ and ${\tt tok}(A'')$ is in state {\tt token}. If $r''=r'''$, then ${\tt tok}(A)$ and ${\tt tok}(A'')$ are both in state {\tt token} at ${\tt home}(A)={\tt home}(A'')$, and thus ${\tt tok}(A)={\tt tok}(A'')$, or otherwise we get a contradiction with Lemma~\ref{lem:notwotok}. If $r'''<r''$, we can also get the guarantee that ${\tt tok}(A)={\tt tok}(A'')$ but it requires a bit more explanation. Precisely, if $r'''<r''$, agent $A''$ is in state {\tt explorer} at ${\tt home}(A)={\tt home}(A'')$ in each round of $[r'''\ldotp\ldotp r'']$. Hence, if no agent transits from state {\tt explorer} to state {\tt searcher} at ${\tt home}(A'')$ in some round of $[r'''\ldotp\ldotp r''-1]$, we know, by the description of state {\tt token}, that ${\tt tok}(A'')$ is in state {\tt token} at ${\tt home}(A'')$ in each round of $[r'''\ldotp\ldotp r'']$. As above, this necessarily means that ${\tt tok}(A)$ and ${\tt tok}(A'')$ are both in state {\tt token} at ${\tt home}(A)={\tt home}(A'')$, and thus ${\tt tok}(A)={\tt tok}(A'')$ by Lemma~\ref{lem:notwotok}. But what if some agent $X$ transits from state {\tt explorer} to state {\tt searcher} at ${\tt home}(A'')$ in some round $t$ of $[r'''\ldotp\ldotp r''-1]$?  In view of the minimality of $r^*$ and lines~\ref{l:5}-\ref{l:afterrepeat2} of Algorithm~\ref{alg:algexplo}, we know that, in round $t$, agent $X$ is with ${\tt tok}(X)$ at ${\tt home}(X)={\tt home}(A'')$ and ${\tt tok}(X)$ is in state {\tt token}. Moreover, assuming, without loss of generality, that $t$ is the first round of $[r'''\ldotp\ldotp r''-1]$ in which an agent transits from state {\tt explorer} to state {\tt searcher} at ${\tt home}(A'')$, ${\tt tok}(A'')$ is also in state {\tt token} at ${\tt home}(X)={\tt home}(A'')$ in round $t$, in view of the description of state {\tt token} and the fact that in each round of $[r'''\ldotp\ldotp t]$ ${\tt tok}(A'')$  is with an agent in state {\tt explorer} that does not declare termination (namely $A''$). Thus, ${\tt tok}(X)={\tt tok}(A'')$ by Lemma~\ref{lem:notwotok}, which implies that $X=A''$ because, by definition, ${\tt tok}(X)$ is the token of exactly one agent, namely $X$. However, by Lemma~\ref{lem:onlyone}, an agent can enter state {\tt explorer} at most once, which means that agent $A''=X$ cannot be in state {\tt explorer} in round $r''>t$. This is a contradiction confirming that, for the first case of our analysis, we must have ${\tt tok}(A)={\tt tok}(A'')$.

By definition, the fact that ${\tt tok}(A)={\tt tok}(A'')$ implies that $A=A''$. However, agent $A''$ decides to leave state {\tt explorer} in round  $r''<r^*$ and agent $A$ is still in this state in round $r^*$, which means that $A''\ne A$ in view of Lemma~\ref{lem:onlyone}. This is a contradiction that shows that the first case of our analysis cannot occur.



The second case to consider is when $\mathcal{E}$ starts by round $r^*$ but is not completed by round $r^*$, and ${\tt tok}(A)$ decides in this round to switch from state {\tt token} to state {\tt searcher} or to declare termination. However, using similar arguments as those in the analysis of the first case, we again reach a contradiction.

As a result, round $r^*$ cannot invalidate property $\Psi(\mathcal{E},A)$. This shows that $\Psi(\mathcal{E},A)$ is necessarily satisfied and concludes the proof of the lemma. 
\end{proof}

The next lemma highlights some relationship between a token and its explorer.

\begin{lemma}
\label{lem:transit}
Agent $A$ transits from state {\tt explorer} to state {\tt searcher} in a round $r$ if, and only if, ${\tt tok}(A)$ transits from state {\tt token} to state {\tt searcher} in round $r$.
\end{lemma}

\begin{proof}
From the round when an agent has completed its last execution of line~\ref{l:5} of Algorithm~\ref{alg:algexplo} to the round when it executes line~\ref{l:afterrepeat2} or \ref{l:afterrepeat3} of Algorithm~\ref{alg:algexplo} (this round included), the agent does not move and is always in state {\tt explorer}. Hence, by Lemma~\ref{lem:exp}, we have the following claim.

\begin{claim}
\label{claim27}
When an agent $B$ in state {\tt explorer} declares termination or switches to state {\tt searcher}, it is necessarily at ${\tt home}(B)$. 
\end{claim}

Below, we prove both directions of the equivalence separately. We start with the forward direction of the equivalence. Thus, assume that an agent $A$ transits from state {\tt explorer} to state {\tt searcher} in a round $r$. We need to show that ${\tt tok}(A)$ transits from state {\tt token} to state {\tt searcher} in round $r$. Note that, according to Claim~\ref{claim27}, agent $A$ is in state {\tt explorer} at ${\tt home}(A)$ in round $r$.

There is a case that can be handled easily. It is the one where, in each round, if an agent $B$ is in state {\tt explorer} at ${\tt home}(B)$, then ${\tt tok}(B)$ is in state {\tt token} at ${\tt home}(B)$ in the same round. Indeed, this case immediately implies that ${\tt tok}(A)$ is in state {\tt token} at ${\tt home}(A)$ with agent $A$ in round $r$. Consequently, by the transition rule given in the description of state {\tt token}, we know that agent ${\tt tok}(A)$ transits from state {\tt token} to state {\tt searcher} in round $r$.



So, consider the complementary case in which there exists a round where an agent $B$ is in state {\tt explorer} at ${\tt home}(B)$, but ${\tt tok}(B)$ is not in state {\tt token} at ${\tt home}(B)$ in this round. Let $r'$ be the first round in which this occurs. In view of the description of state {\tt token}, it turns out that there is a round $r''<r'$ in which the following two conditions are satisfied: $(1)$ ${\tt tok}(B)$ (resp. some agent $D$) is in state {\tt token} (resp. {\tt explorer}) at ${\tt home}(B)$ and $(2)$ agent $D$ declares termination or switches to state {\tt searcher}. According to Claim~\ref{claim27}, we know that agent $D$ is at ${\tt home}(D)$ in round $r''$. Therefore ${\tt home}(B)={\tt home}(D)$. Moreover, by the definition of round $r'$ and the fact that $r''<r'$, we have the guarantee that ${\tt tok}(D)$ is in state {\tt token} at ${\tt home}(B)={\tt home}(D)$ in round $r''$. Given that ${\tt tok}(B)$ and ${\tt tok}(D)$ are both in state {\tt token} at ${\tt home}(B)={\tt home}(D)$ in round $r''$, we necessarily have  ${\tt tok}(B)={\tt tok}(D)$ by Lemma~\ref{lem:notwotok}. This implies that $B=D$. However, since, in round $r''$, agent $D$ is in state {\tt explorer} and declares termination or switches to state {\tt searcher}, we know that $B\ne D$ in view of Lemma~\ref{lem:onlyone} and the fact that an agent declaring termination in round $r''$ is considered to have no state in any subsequent round. This is a contradiction, which concludes the proof of the forward direction of the equivalence.


Now, consider the backward direction. Assume for the sake of contradiction that ${\tt tok}(A)$ transits from state {\tt token} to state {\tt searcher} in round $r$, but agent $A$ does not transit from state {\tt explorer} to state {\tt searcher} in a round $r$. In view of the description of state {\tt token} and the fact that while in state {\tt token} an agent does not move, we know that, in round $r$, agent ${\tt tok}(A)$ is at ${\tt home}(A)$ with an agent $B\ne A$ that transits from state {\tt explorer} to state {\tt searcher}. According to Claim~\ref{claim27}, we have ${\tt home}(A)={\tt home}(B)$. Moreover, since $B$ transits from state {\tt explorer} to state {\tt searcher} in round $r$, we know that, during its last iteration of the repeat loop of Algorithm~\ref{alg:algexplo}, the execution of ${\tt EST}^+$ returns $(\true,*,*)$ and line~\ref{l:if2} of Algorithm~\ref{alg:algexplo} is not executed. Thus, an execution of ${\tt EST}^+$ is completed by agent $B$ in round $r$, which means, by Lemma~\ref{lem:exp}, that ${\tt tok}(B)$ is in state {\tt token} at ${\tt home}(A)={\tt home}(B)$ in round $r$. From this point, we can again reach a contradiction by showing both that $A=B$ and $A\ne B$, following similar arguments to those used in the analysis of the forward direction. This concludes the proof of the backward direction of the equivalence, and thus the proof of the lemma.
\end{proof}

In the sequel, we denote by $r_1$ the first round in which at least one agent enters state {\tt explorer}. This round, which is involved in the next lemma, exists by Lemmas~\ref{lem:tokappears} and~\ref{lem:notwotok}.

\begin{lemma}
\label{claim2}
Let $EXP$ be the (non-empty) set of agents that enters state {\tt explorer} in round $r_1$. There exists an agent of $EXP$ that never switches from state {\tt explorer} to state {\tt searcher}.
\end{lemma}

\begin{proof}
 Assume, for the sake of contradiction, that the claim does not hold. For every agent $A$ of $EXP$, denote by $r_A$ the round in which it decides to switch from state {\tt explorer} to state {\tt searcher} (this round is unique by Lemma~\ref{lem:onlyone}). Let $r^*$ be the latest of these rounds and let $A'$ be an agent of $EXP$ such that for every $A\in EXP$, $M_{r^*}({\tt tok}(A))$ is smaller than or equal to $M_{r^*}({\tt tok}(A'))$. 
 
 

By the first property of Lemma~\ref{lem:exp}, we know that at the beginning of every execution of function ${\tt EST}^+$ by agent $A'$ the input of the function is well defined and corresponds to the current memory of ${\tt tok}(A')$. Consequently, the description of state {\tt explorer} and the fact that $A'$ decides to switch from state {\tt explorer} to state {\tt searcher} in round $r_{A'}$ implies that, in some round $r_1\leq r\leq r_{A'}$, there is an agent $B$ in state {\tt token} that has either $(1)$ a higher seniority than $A'$ or $(2)$ the same seniority as $A'$ but $M_{r}({\tt tok}(A'))\prec M_{r}(B)$.

By the definition of $EXP$, agent $A'$ is among the earliest to enter state {\tt explorer}, which means that the first case cannot occur. Therefore, the second case must occur. As a result, agent $B$ enters state {\tt token} in round $r_1$ and thus, there is an agent $B'$ such that $B={\tt tok}(B')$ (or equivalently such that $B'={\tt exp}(B)$) that enters state {\tt explorer} in round $r_1$. This implies that agent $B'$ belongs to $EXP$. Moreover, in view of Lemma~\ref{lem:transit} and the fact that $B$ enters state {\tt token} at most once (cf. Lemma~\ref{lem:onlyone}), we know that $r_{B'}\geq r$. Thus, in view of Lemma~\ref{lem:prec}, the inequality $M_{r}({\tt tok}(A'))\prec M_{r}({\tt tok}(B'))$ implies that $M_{r^*}({\tt tok}(A'))\prec M_{r^*}({\tt tok}(B'))$. From this and the fact that $B'\in EXP$, we get a contradiction with the definition of $A'$. Hence, the claim is necessarily true.
\end{proof}

Until now, we have essentially established properties concerning states {\tt token} and {\tt explorer}. The following lemma introduces some properties related to state {\tt shadow}.

\begin{lemma}
\label{lem:guide}
Let $A$ be an agent in state {\tt shadow} in a round $r$ at a node $v$. Agent $A$ has exactly one guide $B$ in round $r$. Moreover, in this round, agent $B$ is at node $v$, can be unambiguously identified by $A$ and is either in state {\tt searcher} or {\tt token}.
\end{lemma}

\begin{proof}
Let $r'$ be the round in which agent $A$ enters state {\tt shadow} and let $v'$ be the node occupied by $A$ in round $r'$. Note that, by Lemma~\ref{lem:onlyone}, $r'$ is unique and agent $A$ is in state {\tt shadow} from round $r'$ to round $r$ included.

First assume that $r=r'$. Under this assumption, we have $v=v'$. According to Algorithm~$\mathcal{HG}$, agent $A$ decides to transit to state {\tt shadow}, at node $v'$ in round $r'-1$, from either $(1)$ state {\tt cruiser} or $(2)$ state {\tt searcher}. In the first case, agent $A$ selects as its guide an agent that enters state {\tt token} at node $v'$ in round $r'$. In view of Lemma~\ref{lem:notwotok}, there is at most one agent that can be in state {\tt token} at node $v'$ in round $r'$, which implies that the lemma holds in the first case. In the second case, agent $A$ selects as its guide an agent $B$ that is in state {\tt token} in round $r'-1$ at node $v'$ (cf. the description of state {\tt searcher}) and that is in state {\tt token} or {\tt searcher} at node $v'$ in round $r'$ (cf. the description of state {\tt token}). By Lemma~\ref{lem:notwotok}, there is at most one agent that is in state {\tt token} at node $v'$ in round $r'-1$. Moreover, this means that, among the agents occupying $v'$ in round $r'$, $B$ is the only agent that was in state {\tt token} at this node in the previous round. Hence, the lemma also holds in the second case, and, by extension, when $r=r'$.

In light of the first part of the proof, we know there exists a non negative integer $i$ such that the lemma holds if $r=r'+i$. As a result, to show that the lemma holds when $r>r'$, it is enough to show that the lemma holds if $r=r'+i+1$. So, we now assume that $r=r'+i+1$ and we denote by $v''$ the node occupied by $A$ in round $r'+i$. 

At the beginning of the proof, we observed that agent $A$ is in state {\tt shadow} from round $r'$ to round $r$ included. Since $r=r'+i+1>r'$,  this means that agent $A$ is necessarily in state {\tt shadow} in round $r'+i$. Hence, by the definition of $i$, we know that, in round $r'+i$, the agent $B$ playing the role of $A$'s guide is either in state {\tt searcher} or {\tt token}, occupies node $v''$ and can be unambiguously identified by $A$. Note that if agent $B$ is in state {\tt searcher} in round $r'+i$, then it is in state {\tt searcher} or {\tt shadow} in round $r'+i+1$. Below, we consider two complementary cases.

First, consider the case where agent $B$ is in state {\tt searcher} in rounds $r'+i$ and $r'+i+1$. According to the description of state {\tt shadow}, in round $r'+i+1$, agent $B$ is at node $v$ and remains the unique guide of $A$. Since, for every integer $N\geq2$, an execution of procedure ${\tt EXPLO}(N)$ contains no waiting periods, the nodes occupied by $B$ in rounds $r'+i$ and $r'+i+1$ are different, i.e., $v\ne v''$. Note that, by Lemma~\ref{lem:diff}, we have $M_B(r'+i)\ne M_X(r'+i)$ for every agent $X\ne B$ at node $v''$ in round $r'+i$. Moreover, among the agents at node $v$ in round $r'+i+1$, only those that were at node $v''$ in round $r'+i$ enter $v$ by port ${\tt port}(v'',v)$ in round $r'+i+1$. As a result, agent $A$ can unambiguously identify $B$ in round $r'+i+1$, which proves the lemma in this case.

Now, consider the case where agent $B$ is either $(1)$ in state {\tt searcher} in round $r'+i$ and in state {\tt shadow} in round $r'+i+1$ or $(2)$ in state {\tt token} in round $r'+i$. In view of Algorithm~$\mathcal{HG}$, in round $r'+i+1$, $A$'s guide is necessarily an agent that was in state {\tt token} at node $v''$ in round $r'+i$ (it corresponds to the guide of $B$ if $B$ is in state {\tt shadow} in round $r'+i+1$, or to $B$ itself otherwise). Hence, using similar arguments to those above (when assuming $r=r'$ and $A$'s guide is in state {\tt token} in round $r-1$), we can prove that the lemma holds in the second case as well. This completes the proof of the lemma.   
\end{proof}

The next lemma pinpoints a specific situation in which all nodes of the underlying graph are visited at least once. In particular, it involves function ${\tt EST}^+$ and the notion of ``admitted node'' (cf. the description of procedure ${\tt EST}$ in Section~\ref{sec:preli}).

\begin{lemma}
\label{lem:toutexplo}
Let $A$ be an agent in state {\tt explorer} that starts an execution of ${\tt EST}^+$ in a round $r$. Let $\mathcal{X}$ be the set of all agents $B$ in state {\tt explorer} in round $r$ such that $M_{{\tt tok}(A)}(r)=M_{{\tt tok}(B)}(r)$ ($\mathcal{X}$ includes agent $A$). Every agent of $\mathcal{X}$ starts an execution of ${\tt EST}^+$ in round $r$. Moreover, each node of the underlying graph $G$ is visited and admitted by at least one agent of $\mathcal{X}$ during its execution of ${\tt EST}^+$ started in round $r$.
\end{lemma}

\begin{proof}
We first want to prove that every agent of $\mathcal{X}$ starts an execution of ${\tt EST}^+$ in round $r$. By assumption, it is the case for agent $A$. So, consider any agent $A'\ne A$ from $\mathcal{X}$. By Lemma~\ref{lem:exp}, in round $r$, ${\tt tok}(A)$ is in state {\tt token} and occupies ${\tt home}(A)$ with agent $A$. Moreover, we know that while in state {\tt token}, an agent does not move. Hence, since $M_{{\tt tok}(A)}(r)=M_{{\tt tok}(A')}(r)$, we know that, in round $r$, ${\tt tok}(A')$ is in state {\tt token} at ${\tt home}(A')$ with an agent $Y$ in state {\tt explorer} that has the same seniority as ${\tt tok}(A')$ (and thus as $A'$) and that starts an execution of ${\tt EST}^+$. Note that we necessarily have ${\tt home}(A')={\tt home}(Y)$ in view of Lemma~\ref{lem:exp}. Consequently, agents $Y$ and $A'$ enter state {\tt explorer} at node ${\tt home}(A')$ in the round when agent ${\tt tok}(A')$ enters state {\tt token} at ${\tt home}(A')$. By Lemma~\ref{lem:notwotok}, this necessarily means that $Y=A'$, which proves that every agent of $\mathcal{X}$ indeed starts an execution of ${\tt EST}^+$ in round $r$.

Now, it remains to show that each node of the underlying graph $G$ is visited and admitted by at least one agent of $\mathcal{X}$ during its execution of ${\tt EST}^+$ started in round $r$. As mentioned earlier, every agent of $\mathcal{X}$ is at its home in round $r$. This implies that during its execution of ${\tt EST}^+$ started in this round, each agent $B$ of $\mathcal{X}$ constructs a tree $T_B$ rooted at a node corresponding to ${\tt home}(B)$. Although, in this execution, agent $B$ can sometimes confuse its token with the token of another agent of $\mathcal{X}$, this can never lead agent $B$ to add a ``wrong'' path in $T_B$. Precisely, at each stage of the construction of $T_B$, each path from the root of $T_B$ to any of its node always corresponds to a path from ${\tt home}(B)$ to some node in $G$. This particularly implies that when admitting a node $x$ of $G$ (and thus adding a corresponding node to $T_B$), agent $B$ is necessarily located at node $x$. Hence, for the purpose of the rest of this proof, it is enough to show that each node of the underlying graph $G$ is admitted by at least one agent of $\mathcal{X}$ during its execution of ${\tt EST}^+$ started in round $r$. Assume, by contradiction, it is not the case for some node $v$ of $G$.


Let $\mathcal{H}$ be the set of homes of all the agents from $\mathcal{X}$ and let $\pi=(p_1,q_1,p_2,q_2,\ldots,p_k,q_k)$ be the lexicographically smallest shortest path from a node of $\mathcal{H}$ to node $v$. Let $C$ be an agent of $\mathcal{X}$ such that $\pi$ is a path from ${\tt home}(C)$ to $v$ (it is not needed here, but one can show that $C$ is unique, using Lemma~\ref{lem:notwotok}). Finally, let $\mathcal{E}_C$ be the execution of ${\tt EST}^+$ started by agent $C$ in round $r$. As mentioned at the beginning of this proof, in round $r$, agent ${\tt tok}(C)$ is in state {\tt token} and occupies node ${\tt home}(C)$ with agent $C$. 

During $\mathcal{E}_C$, agent $C$ constructs a tree $T_C$ rooted at a node corresponding to ${\tt home}(C)$. According to procedure {\tt EST}, when admitting a node $x$ of $G$, agent $C$ only attaches a corresponding node $c(x)$ to a node of $T_C$ by an edge with two port numbers, so that the the simple path from the root of $T_C$ to $c(x)$ is a path from ${\tt home}(C)$ to $x$. Moreover, adding a node to $T_C$ occurs only when admitting a node of $G$. Hence, since $v$ is not admitted by $C$ in $\mathcal{E}_C$, we know that, at the end of $\mathcal{E}_C$, there exists an integer $0\leq i <k$ such that $\pi[1,2i]$ is a simple path in $T_C$ from its root to some of its node, while $\pi[1,2(i+1)]$ is not (as defined in Section~\ref{sec:preli}, $\pi[1,0]$ is considered to be the empty path). According to Lemma~\ref{lem:exp}, $\mathcal{E}_C$ eventually terminates. Hence, in view of the description of procedure ${\tt EST}$, agent $C$ eventually moves from ${\tt home}(C)$ to some node $u$ by following path $\pi[1,2i]$, which leads to some node $c(u)$ from the root of $T_C$, and then starts the process of $c(u)$. Denote by $Pr$ this process. During $Pr$, agent $C$ checks each neighbor of $u$ in order to determine whether admitting the neighbor or not. In particular, it checks ${\tt succ}(u,p_{i+1})$ which, by the definition of $i$, is necessarily not admitted during $Pr$. Let $u'={\tt succ}(u,p_{i+1})$. Since, by the first property of Lemma~\ref{lem:exp}, the input given to function ${\tt EST}^+$ in $\mathcal{E}_C$ is $M_{{\tt tok}(C)}(r)$, it follows from the descriptions of ${\tt EST}$ and ${\tt EST}^+$ that there exists a simple path $\varphi$ from node $u'$ satisfying the following two conditions:

\begin{itemize}
\item {Condition~1.} When agent $B$ starts checking $u'$ in $Pr$, $\varphi$ is a simple path in $T_C$ from some node $w$ to the root of $T_C$. 
\item {Condition~2.} For some agent $Z$, agent $C$ encounters an agent ${\tt tok}(Z)$ in state {\tt token} such that $M_{{\tt tok}(C)}(r)=M_{{\tt tok}(Z)}(r)$, right after following path $\varphi$ from node $u'$ during $Pr$.
\end{itemize}

Recall that, in round $r$, agent ${\tt tok}(C)$ is in state {\tt token} and occupies node ${\tt home}(C)$ with agent $C$. Also recall that, while in state {\tt token}, an agent does not move. Hence, the fact that $M_{{\tt tok}(C)}(r)=M_{{\tt tok}(Z)}(r)$ implies that, in round $r$, ${\tt tok}(Z)$ is in state {\tt token} at ${\tt home}(Z)$ with an agent in state {\tt explorer} that has the same seniority as ${\tt tok}(Z)$ and that starts an execution of ${\tt EST}^+$. Using the same reasonning as at the beginning of this proof, we can prove that this agent is necessarily $Z$. As a result, since $M_{{\tt tok}(A)}(r)=M_{{\tt tok}(C)}(r)=M_{{\tt tok}(Z)}(r)$, we have $Z\in \mathcal{X}$ and ${\tt home}(Z)\in\mathcal{H}$.

According to procedure {\tt EST}, when agent $C$ finishes the process of a node of $T_C$ in $\mathcal{E}_C$, the next node to process, if any, is the one reached by following the lexicographically smallest shortest path from the root of $T_C$ to an unprocessed node of $T_C$. Moreover, during the process of a node $c(w)$ of $T_C$, corresponding to a node $w$ of $G$, the neighbors of $w$ are checked, and thus possibly admitted, following the increasing order of the port numbers at this node. In particular, when admitting a neighbor $w'$ of $w$ during the process of $c(w)$, the agent precisely attaches a node $c(w')$ to $c(w)$ by an edge whose port number at node $c(w)$ is ${\tt port}(w,w')$. This implies that if a subsequent neighbor $w''$ of $w$ is admitted during the same process, a corresponding node $c(w'')$ will be attached to $c(w)$ with an edge whose port number at node $c(w)$ is greater than ${\tt port}(w,w')$. Therefore, when agent $B$ starts checking $u'$ in $Pr$, we have the guarantee that each simple path from the root of $T_C$ to any of its node is either shorter than path $\pi[1,2i]({\tt port}(u,u'),{\tt port}(u',u))$ or of equal length but lexicographically smaller. Note that $\pi[1,2i]({\tt port}(u,u'),{\tt port}(u',u))=\pi[1,2(i+1)]$ and is a path from ${\tt home(C)}$ to $u'$, by the definition of $u$ and $u'$. Hence, in view of Condition~1, the reverse of path $\varphi$, call it $\bar{\varphi}$, is shorter than $\pi[1,2(i+1)]$ or of equal length but lexicographically smaller. By Condition~2 and the fact that ${\tt tok}(Z)$ stays at ${\tt home}(Z)$ while in state {\tt token}, it follows that $\bar{\varphi}(p_{i+2},q_{i+2},p_{i+3},q_{i+3},\ldots,p_k,q_k)$ is a path from a node of $\mathcal{H}$, namely ${\tt home}(Z)$, to node $v$ that is shorter than $\pi$ or of the same length but lexicographically smaller. This contradicts the definition of $\pi$ and thus the existence of $v$, which terminates the proof. 
\end{proof}

Every execution by an agent in state {\tt explorer} of the repeat loop of Algorithm~\ref{alg:algexplo} will be viewed as a sequence of consecutive steps $j= 1,2,3,4,\ldots$ where step $j$ is the part of the execution corresponding to the $j$th iteration of this loop. We will say that two steps, executed by the same agent or not, are identical if they have the same duration and the triple returned by function ${\tt EST}^+$ in the first step is equal to the triple returned by ${\tt EST}^+$ in the second step. The notion of step is used in the following four lemmas.


Note that, in view of Algorithm~\ref{alg:algexplo}, an agent can transit to state {\tt explorer} only from state {\tt cruiser} and, by Lemma~\ref{lem:onlyone}, once an agent leaves state {\tt cruiser}, it cannot return to it. Hence, in the statement of the next lemma, the value $\tau$ is well defined.

\begin{lemma}
\label{lem:stepduration}
Let $s$ be a step completed in some round by an agent $A$ in state {\tt explorer}, in which function ${\tt EST}^+$ returns a triple $(b,\eta,trace)$ such that $b=\false$. Let $\tau$ be the number of rounds that agent $A$ spends in state {\tt cruiser}. The duration $T$ of step $s$ satisfies the following bounds: $(\eta U\tau)^{11\beta}\leq T \leq 2^{13}(nU\tau)^{11\beta}$.
\end{lemma}

\begin{proof}
Since, during step $s$, function ${\tt EST}^+$ returns a triple $(b,\eta,trace)$ such that $b=\false$, we know, in view of line~\ref{l:if} of Algorithm~\ref{alg:algexplo}, that the waiting period at line~\ref{l:if2} of Algorithm~\ref{alg:algexplo} is executed in step $s$. Note that $\beta\geq2, U\geq 2$, $\tau\geq 1$ (the agent spends at least one round in state {\tt cruiser}) and, by Lemma~\ref{lem:exp}, $1\leq\eta\leq n$. This means that the execution time of line~\ref{l:if2} of Algorithm~\ref{alg:algexplo} in step $s$ can be can be lower bounded (resp. upper bounded) by $(\eta U\tau)^{11\beta}$ (resp. $2^{12}(nU\tau)^{11\beta}$). Moreover, according to Lemma~\ref{lem:exp}, an execution of line~\ref{l:5} of Algorithm~\ref{alg:algexplo} lasts at most $8n^5$ rounds. Hence, step $s$ lasts at least (resp. at most) $(\eta U\tau)^{11\beta}$ (resp. $8n^5+2^{12}(nU\tau)^{11\beta}\leq 2^{13}(nU\tau)^{11\beta}$) rounds, which proves the lemma.
\end{proof}

Below is arguably one of the most important lemma of this section. It will serve as the main argument to show that gathering is indeed done when a sequence of $U^{25\beta}$ consecutive identical steps is completed by some agent (cf. the proof of Lemma~\ref{lem:gatheroccurs}).

\begin{lemma}
\label{lem:important}
Let $r$ be the first round, if any, in which a sequence of $U^{25\beta}$ consecutive identical steps is completed by some agent $A$ in state {\tt explorer}. No agent has declared termination before round $r$ and, in this round, agent ${\tt tok}(A)$ (resp. every agent different from $A$ and ${\tt tok}(A)$) is in state {\tt token} (resp. {\tt shadow}).
\end{lemma}

\begin{proof}
Note that an agent in state {\tt explorer} (resp. {\tt token}) declares termination only when it has completed a sequence of $U^{25\beta}$ consecutive identical steps (resp. only when it is with an agent in state {\tt explorer} that declares termination). By the definition of round $r$, this means there is no round $r'<r$ in which an agent in state {\tt explorer} or {\tt token} declares termination. Moreover, while in state {\tt searcher} or {\tt cruiser}, an agent cannot declare termination and an agent in state {\tt shadow} declares termination only if its guide declares termination. However, by Lemma~\ref{lem:guide}, in each round, the guide of an agent in state {\tt shadow} is either in state {\tt searcher} or {\tt token}. Hence, no agent declares termination before round $r$. From this, Lemma~\ref{lem:transit} and the description of state {\tt token}, it particularly follows that if agent ${\tt tok}(A)$ is not in state {\tt token} in round $r$, then agent $A$ transits from state {\tt explorer} to state {\tt searcher} before round $r$, which contradicts Lemma~\ref{lem:onlyone} and the fact that agent $A$ is in state {\tt explorer} in round $r$. 

Consequently, to prove the lemma, it remains to show that, in round $r$, every agent, different from $A$ and ${\tt tok}(A)$, is in state {\tt shadow}. In the rest of this proof, given a step $s$ executed by an agent $X$ in state {\tt explorer}, we will denote by $r_s$ (resp. $r'_s$) the round in which $s$ is started (resp. completed) and by $\mathcal{X}_s$ the set of agents $X'$ in state {\tt explorer} in round $r_s$ such that $M_{{\tt tok}(X)}(r_s)=M_{{\tt tok}(X')}(r_s)$ ($\mathcal{X}_s$ includes agent $X$). We will also say that a sequence of $l$ consecutive (not necessarily identical) steps $S=s_1,s_2,\ldots,s_l$ completed in round $r'_{s_l}$ by an agent in state {\tt explorer} is $l$-perfect if, the next two properties hold:

\begin{enumerate}
\item In step $s_l$, function ${\tt EST}^+$ returns a triple whose the first element is the Boolean value false.
\item For every $1\leq i \leq l$, $\mathcal{X}_{s_i}=\mathcal{X}_{1}$ and a step identical to $s_i$ is completed by each agent of $\mathcal{X}_{s_i}$ in round $r'_{s_i}$.
\end{enumerate}

Note that the first property and lines~\ref{l:repeat} to~\ref{l:afterrepeat2} of Algorithm~\ref{alg:algexplo} immediately imply that function ${\tt EST}^+$ returns a triple whose the first element is the Boolean value false in each step of $S$. Also note that in every round, there cannot be more than $\mu$ non-dormant agents that have the same memory. Hence, since for each agent $X$ there exists at most one agent ${\tt exp}(X)$, it follows that the cardinality of $\mathcal{X}_{s_i}$ is at most $\mu$ for every $1\leq i \leq l$. 

We now proceed with the following two claims.

\begin{claim}
\label{claim24}
Consider any sequence of $k(\mu+1)$ consecutive steps $s_{1},s_{2},\ldots,s_{k(\mu+1)}$ completed in round $r'_{s_{k(\mu+1)}}$ by an agent $X$ in state {\tt explorer}, where $k$ is a positive integer. There exists a positive integer $j \leq k\mu$ such that the subsequence $s_j,s_{j+1},\ldots,s_{j+k-1}$ is $k$-perfect.
\end{claim}

\begin{proofclaim}
Let us first prove that, for every $1\leq i < k(\mu+1)$, if an agent of $\mathcal{X}_{s_i}$ starts in round $r_{s_i}$ a step that will not be identical to $s_i$, then this agent cannot belong to $\mathcal{X}_{s_{i+1}}$. Suppose by contradiction this does not hold. Hence, there exist an integer $1\leq i<k(\mu+1)$ and an agent $X'$ such that $X'$ is in $\mathcal{X}_{s_i}\cap\mathcal{X}_{s_{i+1}}$ while the step started in round $r_{s_i}$ by agent $X'$ is not identical to $s_i$. By Lemma~\ref{lem:prec}, we necessarily have $M_X(r_{s_{i+1}})\ne M_{X'}(r_{s_{i+1}})$. Moreover, by the definition of $\mathcal{X}_{s_{i+1}}$ and Lemma~\ref{lem:exp}, we know that $M_{{\tt tok}(X)}(r_{s_{i+1}})=M_{{\tt tok}(X')}(r_{s_{i+1}})$ and, in round $r_{s_{i+1}}$, agent $X$ is in state {\tt explorer} at ${\tt home}(X)$ with ${\tt tok}(X)$ that is in state {\tt token}. Thus, in round $r_{s_{i+1}}$, ${\tt tok}(X')$ is in state {\tt token} at ${\tt home}(X')$ with an agent $Y$ in state {\tt explorer} that starts a step and such that $M_X(r_{s_{i+1}})=M_{Y}(r_{s_{i+1}})$. This particularly implies that $Y\ne X'$. However from Lemmas~\ref{lem:notwotok} and~\ref{lem:exp}, it follows that ${\tt tok}(X')={\tt tok}(Y)$, which implies that $X'=Y$ as, by definition, ${\tt tok}(X')$ is the token of exactly one agent, namely $X'$. This is a contradiction that proves the implication given at the beginning of the proof of the claim.


In light of this and Lemma~\ref{lem:prec}, we can state that, for every integer $1\leq i <k(\mu+1)$, $\mathcal{X}_{s_{i+1}}\subset \mathcal{X}_{s_{i}}$ if at least one agent of $\mathcal{X}_{s_{i}}$ starts in round $r_{s_i}$ a step that will not be identical to $s_i$, $\mathcal{X}_{s_{i+1}}\subseteq \mathcal{X}_{s_{i}}$ otherwise. Moreover, lines~\ref{l:repeat}-\ref{l:afterrepeat2} of Algorithm~\ref{alg:algexplo} and the fact that the steps $s_{1},s_{2},\ldots,s_{k(\mu+1)}$ are consecutive imply that function ${\tt EST}^+$ necessarily returns a triple whose the first element is the Boolean value false in each of these steps, except possibly in step $s_{k(\mu+1)}$. Hence, if the claim does not hold, it means that for every integer $0\leq m \leq \mu-1$, the cardinality of $\mathcal{X}_{s_{1+km}}$ is smaller than that of $\mathcal{X}_{s_{1+k(m+1)}}$. However, since the cardinality of $\mathcal{X}_{s_{1}}$ is at most $\mu$ (cf. the explanations given just above the statement of the claim), $\mathcal{X}_{s_{1+k\mu}}$ must be empty, which is impossible as $X\in \mathcal{X}_{s_{1+k\mu}}$. Consequently, the claim necessarily holds.
\end{proofclaim}

\begin{claim}
\label{claim34}
Let $k$ be a positive integer and let $S=s_1,s_2,\ldots,s_k$ be a $k$-perfect sequence of steps completed in round $r'_{s_k}$ by an agent $X$ in state {\tt explorer}. For every integer $1\leq i\leq k$, each node is visited by at least one agent of $\mathcal{X}_{s_i}$ in some round $r_{s_i}\leq t \leq r'_{s_i}$. Moreover, if all the steps of $S$ are identical, then the second elements of the triples returned by function ${\tt EST}^+$ in the steps of $S$ all have the same value $\eta$, with $\eta\geq \frac{n}{\mu}$.
\end{claim}

\begin{proofclaim}
From Lemma~\ref{lem:toutexplo} and the fact that $S$ is a $k$-perfect sequence of steps, we immediately get the guarantee that for every integer $1\leq i\leq k$, each node is visited by at least one agent of $\mathcal{X}_{s_i}$ during its execution of ${\tt EST}^+$ within the time interval $[r_{s_i}\ldotp\ldotp r'_{s_i}]$. Moreover, if all the steps of $S$ are identical, it is obvious that the second elements of the triples returned by function ${\tt EST}^+$ in the steps of $S$ all have the same value $\eta$. Therefore, to prove the claim, it is enough to show that the second element of the triple returned by function $EST^+$ in $s_1$, call it $\eta_1$, is at least $\frac{n}{\mu}$.

By Lemma~\ref{lem:toutexplo}, we know that each node is admitted by at least one agent of $\mathcal{X}_{s_1}$ during its execution of ${\tt EST}^+$ within the time interval $[r_{s_1}\ldotp\ldotp r'_{s_1}]$. Each time a node is admitted during an execution of ${\tt EST}^+$, a new node is added to the BFS tree under construction in this execution. Hence, the sum of the orders of the trees built by the agents of $\mathcal{X}_{s_1}$ during their execution of ${\tt EST}^+$ in the time interval $[r_{s_1}\ldotp\ldotp r'_{s_1}]$ is at least $n$. Note that these trees all have the same order $\eta_1$ as a step identical to $s_1$ is completed by each agent of $\mathcal{X}_{s_i}$ in round $r'_{s_1}$. Also recall that the cardinality of $\mathcal{X}_{s_1}$ is at most $\mu$ (cf. the explanations given just before Claim~\ref{claim24}). As a result,   $\eta_1\geq \frac{n}{\mu}$, which concludes the proof of this claim. 
\end{proofclaim}
~\\
In the sequel, we denote by $r^*$ the round in which agent $A$ (resp. ${\tt tok}(A)$) enters state {\tt explorer} (resp. {\tt token}) at ${\tt home}(A)$ and we denote by $\tau$ the number of rounds spent by agent $A$ in state {\tt cruiser} (we necessarily have $\tau\geq 1$). Since we established above that ${\tt tok}(A)$ is in state {\tt token} in round $r$, we know, in view of the description of state {\tt token} and Lemma~\ref{lem:onlyone}, that agent ${\tt tok}(A)$ is in state {\tt token} at ${\tt home}(A)$ from round $r^*$ to round $r$ included. 

Moreover, observe that $U^{25\beta}\geq (U+1)U^{23\beta} \geq (\mu +1)U^{23\beta}$ because $U\geq 2$ and $U\geq \mu$. Also recall that a sequence of $U^{25\beta}$ consecutive identical steps is completed in round $r$ by agent $A$. Hence, in view of Claim~\ref{claim24}, we get the guarantee that a $U^{23\beta}$-perfect sequence of identical steps $P=\rho_1,\rho_2,\ldots,\rho_{U^{23\beta}}$ is completed in round $r'_{\rho_{U^{23\beta}}}\leq r$ by agent $A$ (while in state {\tt explorer}).

By Claim~\ref{claim34}, we know that when $\rho_1$ is completed, all nodes of the graph have been visited at least once. Therefore, we have the following claim.

\begin{claim}
\label{claim35}
Every agent enters state {\tt cruiser} by round $r'_{\rho_{1}}$.
\end{claim}

With the claim below, we give a round from which we will no longer see an agent in state {\tt cruiser}.

\begin{claim}
\label{claim35bis}
From round $r'_{\rho_{2}}$ onwards, no agent will be in state {\tt cruiser}.
\end{claim}

\begin{proofclaim}
Assume by contradiction that the claim does not hold. By Lemma~\ref{lem:onlyone} and Claim~\ref{claim35}, this means there exists an agent $X$ that wakes up in some round $t\leq r'_{\rho_{1}}$ and that is in state {\tt cruiser} from round $t$ to at least round $r'_{\rho_{2}}$ included.
Note that by Lemma~\ref{lem:stepduration} and Claim~\ref{claim34}, $r'_{\rho_{2}}-r'_{\rho_{1}}$ is at least equal to $(\frac{n}{\mu}U\tau)^{11\beta}\geq (n\tau)^{11\beta}\geq 1$.

In view of Lemma~\ref{lem:diff} and the transition rules from state {\tt cruiser}, we know that when an agent in state {\tt cruiser} meets an agent in state {\tt cruiser} or {\tt token} in some round, it enters state {\tt token}, {\tt explorer} or {\tt shadow} in the next round. Thus, from round $t$ to round $r'_{\rho_{2}}-1$ included, agent $X$ never encounters an agent in state {\tt cruiser} or {\tt token}, otherwise we get a contradiction with the definition of $X$.

While in state {\tt cruiser}, agent $X$ works in phases $i= 1,2,\ldots$ where phase $i$ consists of an execution of procedure ${\tt EXPLO}(2^i)$ followed by an execution, for $2^i$ rounds, of procedure ${\tt TZ}(\ell_X)$. An entire execution of the $i$th phase lasts $2^{i\beta}+2^i$ rounds and, by the properties of procedure ${\tt EXPLO}$, allows to visit each node of the underlying graph at least once if $2^i\geq n$.

Clearly, when agent $A$ wakes up, agent $X$ cannot have entirely completed the first $\lceil \log n \rceil$ phases of state {\tt cruiser}, because otherwise, by round $r'_{\rho{1}}$, $X$ wakes up an agent that immediately enters state {\tt cruiser} (namely $A$), which is a contradiction. Hence, in view of the duration of a phase in state {\tt cruiser}, we know that agent $X$ cannot have started by round $r^*$ the $(\lceil \log n \rceil + \lceil \log \tau \rceil+1)$th phase of state {\tt cruiser}. However, $\lceil \log n \rceil + \lceil \log \tau \rceil+1\leq 4\lceil \log(n\tau) \rceil +1$, and an entire execution of the first $4\lceil \log(n\tau) \rceil +1$ phases of state {\tt cruiser} lasts at most $\sum_{i=1}^{4\lceil \log(n\tau) \rceil +1}(2^{i\beta}+2^i) < 2^{5\beta+2}(n\tau)^{4\beta}$ rounds, which is at most $\frac{2^{5\beta+2}}{(n\tau)^{7\beta}} (n\tau)^{11\beta}\leq (n\tau)^{11\beta}$ rounds because $n\tau\geq 2$ and $\beta\geq 2$. Since $(n\tau)^{11\beta}\leq r'_{\rho_{2}}-r'_{\rho_{1}}$, it follows that an entire execution of the $(4\lceil \log(n\tau) \rceil +1)$th phase of state {\tt cruiser} is started (resp. completed) by agent $A$ after round $r^*$ (resp. before round $r'_{\rho_{2}}$). This execution allows agent $X$ to visit each node of the graph and to meet an agent in state {\tt token} before round $r'_{\rho_{2}}$ (as ${\tt tok}(A)$ is in state {\tt token} at ${\tt home}(A)$ from round $r^*$ to round $r$ included). This is a contradiction, which concludes the proof of the claim.
\end{proofclaim}

The next claim gives a restriction on the set of agents that can be in state {\tt explorer} in round $r'_{\rho_{U^{22\beta}}}$ and after.

\begin{claim}
\label{claim36}
From round $r'_{\rho_{U^{22\beta}}}$ onwards, no agent outside of $\mathcal{X}_{\rho_{U^{22\beta}}}$ will be in state {\tt explorer}.
\end{claim}

\begin{proofclaim}
Assume by contradiction that the claim does not hold. This means there exists an agent $Y$ outside of $\mathcal{X}_{\rho_{U^{22\beta}}}$ that is in state {\tt explorer} in round $r'_{\rho_{U^{22\beta}}}$ or after. Since an agent can transit to state {\tt explorer} only from state {\tt cruiser}, Claim~\ref{claim35bis} implies that agent $Y$ is in state {\tt explorer} from round $r'_{\rho_{2}}$ to round $r'_{\rho_{U^{22\beta}}}$ included, and thus from round $r'_{\rho_{2}}$ to round $r_{\rho_{U^{22\beta}}}$.

We know, by Lemma~\ref{lem:stepduration} and Claim~\ref{claim34}, that $r_{\rho_{U^{22\beta}}}-r'_{\rho_{2}}\geq  (U^{22\beta}-3)(\frac{n}{\mu}U\tau)^{11\beta}\geq (U^{22\beta}-3)(n\tau)^{11\beta}$. Moreover, we have $(U^{22\beta}-3)(n\tau)^{11\beta}\geq U^{11\beta}(nU\tau)^{11\beta}-3(n\tau)^{11\beta}$ which is at least $U^9 2^{13}(nU\tau)^{11\beta}-3(n\tau)^{11\beta}\geq (\mu+1) 2^{13}(nU\tau)^{11\beta} +2^{13}(nU\tau)^{11\beta}$. Hence, Claim~\ref{claim24} and Lemma~\ref{lem:stepduration} guarantee that agent $Y$ executes a $k$-perfect sequence of steps $c_1,c',\ldots,c_k$ for some $k\geq1$, with $r_{c_1}\geq r'_{\rho_{2}}$ and  $r'_{c_k}\leq r_{\rho_{U^{22\beta}}}$. In view of the definition of such a sequence and Claim~\ref{claim34}, a step identical to $c_1$ is completed in round $r'_{c_1}$ by every agent of $\mathcal{X}_{c_1}$, ${\tt tok}(A)$ is visited by some agent of $\mathcal{X}_{c_1}$ in some round of $[r_{c_1}\ldotp\ldotp r'_{c_1}]$ and function ${\tt EST}^+$ returns a triple whose the first element is the Boolean value false in each step completed by an agent of $\mathcal{X}_{c_1}$ in round $r'_{c_1}$. Therefore, by the description of state {\tt explorer} and Lemma~\ref{lem:exp}, we have one of the following properties in round $r_{c_1}$: either $(1)$ the seniority of ${\tt tok}(Y)$ is greater than that of ${\tt tok}(A)$ or $(2)$ their seniorities are equal but the memory of ${\tt tok}(Y)$ is larger than or equal to the memory of ${\tt tok}(A)$.

Since agent $Y$ is in state {\tt explorer} from round $r'_{\rho_{2}}$ to round $r'_{\rho_{U^{22\beta}}}$ included, we know that agent ${\tt tok}(Y)$ is in state {\tt token} at ${\tt home}(Y)$ from round $r_{c_1}$ to round $r'_{\rho_{U^{22\beta}}}$ included, in view of Lemma~\ref{lem:transit} and the fact that no agent has declared termination before round $r$. Moreover, $r_{c_1}\leq r_{\rho_{U^{22\beta}}}$ because step $c_k$ is completed in round $r'_{c_k}$ and $r'_{c_k}\leq r_{\rho_{U^{22\beta}}}$. Thus, Claim~\ref{claim34} and the fact that $P=\rho_1,\rho_2,\ldots,\rho_{U^{23\beta}}$ is a $U^{23\beta}$-perfect sequence imply that ${\tt tok}(Y)$ is visited by some agent of $\mathcal{X}_{\rho_{U^{22\beta}}}$ in some round of $[r_{\rho_{U^{22\beta}}}\ldotp\ldotp r'_{\rho_{U^{22\beta}}}]$ and function ${\tt EST}^+$ returns a triple whose the first element is the Boolean value false in each step completed by an agent of $\mathcal{X}_{\rho_{U^{22\beta}}}$ in round $r'_{\rho_{U^{22\beta}}}$. Therefore, by the description of state {\tt explorer} and Lemma~\ref{lem:exp}, we have one of the following properties in round $r_{\rho_{U^{22\beta}}}$: either $(1)$ the seniority of ${\tt tok}(A)$ is greater than that of ${\tt tok}(Y)$ or $(2)$ their seniorities are equal but the memory of ${\tt tok}(A)$ is larger than or equal to the memory of ${\tt tok}(Y)$. By Lemma~\ref{lem:prec} and the fact that $r_{\rho_{U^{22\beta}}}\geq r_{c_1}$, it means that, in round $r_{c_1}$, either $(1)$ the seniority of ${\tt tok}(A)$ is greater than that of ${\tt tok}(Y)$ or $(2)$ their seniorities are equal but the memory of ${\tt tok}(A)$ is larger than or equal to the memory of ${\tt tok}(Y)$. According to what we established at the end of the previous paragraph, it necessarily follows that $M_{{\tt tok}(A)}(r_{c_1})=M_{{\tt tok}(Y)}(r_{c_1})$, and thus $M_{{\tt tok}(A)}(r_{\rho_3})=M_{{\tt tok}(Y)}(r_{\rho_3})$ as $r_{\rho_3}=r'_{\rho_2}\leq r_{c_1}$.

Recall that agent $Y$ is in state {\tt explorer} from round $r'_{\rho_{2}}$ to round $r'_{\rho_{U^{22\beta}}}$ included, which means it is in this state in round $r_{\rho_3}$. Consequently, Lemma~\ref{lem:toutexplo} and the fact that $M_{{\tt tok}(A)}(r_{\rho_3})=M_{{\tt tok}(Y)}(r_{\rho_3})$ imply that $Y$ starts an execution of ${\tt EST}^+$, and thus a step, in round $r_{\rho_3}$. Hence, $Y\in \mathcal{X}_{\rho_3}$. However, $\mathcal{X}_{\rho_3}=\mathcal{X}_{\rho_{U^{22\beta}}}$ as $P$ is a $U^{23\beta}$-perfect sequence. This contradicts the fact that $Y\notin \mathcal{X}_{\rho_{U^{22\beta}}}$, which proves the claim. 

\end{proofclaim}

Before concluding the proof of the lemma, we still need two more claims that are proven below.

\begin{claim}
\label{claim37}
In round $r_{\rho_{U^{23\beta}}}$, no agent is in state {\tt searcher}.
\end{claim}

\begin{proofclaim}
Assume by contradiction that there is an agent $X$ in state {\tt searcher} in round $r_{\rho_{U^{23\beta}}}$. Note that, in view of Algorithm~\ref{alg:algexplo}, an agent can transit to state {\tt searcher} only from state {\tt explorer} or {\tt token}. Moreover, by Lemma~\ref{lem:notwotok}, if an agent $Y$ enters state {\tt token} in some round then, in the same round, there exists an agent $Y'$ that enters state {\tt explorer} for which agent $Y$ becomes the token. Thus, by Lemma~\ref{lem:transit}, we know that an agent transits to state {\tt searcher} in a round $t$ iff an agent in state {\tt explorer} transits to state {\tt searcher} in round $t$.

By Lemma~\ref{lem:onlyone} and the fact that $P$ is a $U^{23\beta}$-perfect sequence of steps, we know that the agents of $\mathcal{X}_{\rho_{U^{22\beta}}}=\mathcal{X}_{\rho_{U^{23\beta}}}$ are in state {\tt explorer} from round $r_{\rho_{U^{22\beta}}}$ to round $r_{\rho_{U^{23\beta}}}$ included. Moreover, Claim~\ref{claim36} implies that every agent outside of $r_{\rho_{U^{22\beta}}}$ cannot be in state {\tt explorer} in any round of $[r_{\rho_{U^{22\beta}}}\ldotp\ldotp r_{\rho_{U^{23\beta}}}]$. Hence, in each round of $[r_{\rho_{U^{22\beta}}}\ldotp\ldotp r_{\rho_{U^{23\beta}}}-1]$, no agent in state {\tt explorer} transits to state {\tt searcher}. Since we stated above that an agent transits to state {\tt searcher} in a round $t$ iff an agent in state {\tt explorer} transits to state {\tt searcher} in round $t$, it follows that agent $X$ is in state {\tt searcher} from some round $t'\leq r_{\rho_{U^{22\beta}}}$ to round $r_{\rho_{U^{23\beta}}}$ included.

While in state {\tt searcher}, agent $X$ works in phases $i=1,2,\ldots$ where phase $i$ consists of an execution of procedure ${\tt EXPLO}(2^i)$. This is interrupted only when agent $X$ meets an agent in state {\tt token}, in which case it transits to state {\tt shadow} in the round of this meeting. For every positive integer $k$, an entire execution of the first $\lceil\log k\rceil$ phases lasts at most $\sum_{i=1}^{\lceil\log k\rceil}2^{i\beta}$, which is upper bounded by $(4k)^{\beta}$. Note that by Lemma~\ref{lem:stepduration} and Claim~\ref{claim34}, $r_{\rho_{U^{23\beta}}}-r_{\rho_{U^{22\beta}}}$ is at least equal to $(\frac{n}{\mu}U\tau)^{11\beta}\geq (n\tau)^{11\beta}$. Hence, since $t'\leq r_{\rho_{U^{22\beta}}}$ and $(n\tau)^{11\beta}> (4n)^{\beta}$, the first $\lceil \log n \rceil$ phases of state {\tt searcher} are started and completed by agent $X$ during a time interval included in $[t'\ldotp\ldotp r_{\rho_{U^{23\beta}}}-1]$. Given that an execution of procedure ${\tt EXPLO}(2^i)$ allows the executing agent to visit at least once every node of the graph if $2^i\geq n$, it follows that agent $X$ visits every node of the graph in the time interval $[t'\ldotp\ldotp r_{\rho_{U^{23\beta}}}-1]$.

Recall that $r_1$ is the first round in which there is at least one agent that enters state {\tt explorer}. By Lemma~\ref{lem:notwotok}, we know that $r_1$ is also the first round in which there is at least one agent that enters state {\tt token}. As mentioned above, agent $X$ can transit to state {\tt searcher} only from state {\tt explorer} or {\tt token}. Thus, by Lemmas~\ref{lem:transit} and~\ref{claim2}, we know that $t'\geq r_1$ and there is an agent that is in state {\tt token} from round $r_1$ to round $r_{\rho_{U^{23\beta}}}$ included. As a result, while in state {\tt searcher}, agent $X$ meets an agent in state {\tt token} in some round $t''$ of $[t'\ldotp\ldotp r_{\rho_{U^{23\beta}}}-1]$ and thus enters state {\tt shadow} in round $t''+1$. This contradicts the fact that agent $X$ is in state {\tt searcher} from some round $t'$ to round $r_{\rho_{U^{23\beta}}}$ included. Therefore, the claim holds.

\end{proofclaim}

\begin{claim}
\label{claim38}
In round $r_{\rho_{U^{23\beta}}}$, every agent in state {\tt token} is the token of an agent of $\mathcal{X}_{\rho_{U^{22\beta}}}$.
\end{claim}

\begin{proofclaim}
Consider any agent $X$ in state {\tt token} in round $r_{\rho_{U^{23\beta}}}$. By definition, we know that when $X$ enters state {\tt token} in some round then, in the same round, there exists an agent called ${\tt exp}(X)$ that enters state {\tt explorer} for which $X$ becomes the token. We also know that states {\tt token} and {\tt explorer} can be left only by declaring termination or by transiting to state {\tt searcher}. Moreover, we proved at the beginning of the proof that no agent declares termination before round $r$ and thus by round $r_{\rho_{U^{23\beta}}}$ as $r_{\rho_{U^{23\beta}}}<r$. Hence, by Lemma~\ref{lem:transit}, agent ${\tt exp}(X)$ is in state {\tt explorer} in round $r_{\rho_{U^{23\beta}}}$, which means it belongs to $\mathcal{X}_{\rho_{U^{22\beta}}}$ by Claim~\ref{claim36}. This concludes the proof of this claim.
\end{proofclaim}
~\\

We are now ready to conclude the proof of the lemma. Lemma~\ref{lem:guide} and Claims~\ref{claim35},~\ref{claim35bis} and~\ref{claim37} imply that, in round $r_{\rho_{U^{23\beta}}}$, every agent that is neither in state {\tt explorer} nor {\tt token} is in state {\tt shadow} and shares its current node with an agent in state {\tt token} that is its guide. Claims~\ref{claim36} and~\ref{claim38}, together with the fact that $P$ is a $U^{23\beta}$-perfect sequence of steps, imply that, in round $r_{\rho_{U^{23\beta}}}$, every agent in state {\tt explorer} (resp. {\tt token}) is an agent of $\mathcal{X}_{\rho_{U^{22\beta}}}=\mathcal{X}_{\rho_{U^{23\beta}}}$ (resp. is the token of an agent of $\mathcal{X}_{\rho_{U^{22\beta}}}=\mathcal{X}_{\rho_{U^{23\beta}}}$). Hence, in view of Lemma~\ref{lem:exp}, we know that in round $r_{\rho_{U^{23\beta}}}$, every agent is at a node occupied by the token of an agent of $\mathcal{X}_{\rho_{U^{23\beta}}}$. By the definition of $\mathcal{X}_{\rho_{U^{23\beta}}}$, the tokens of the agents of $\mathcal{X}_{\rho_{U^{23\beta}}}$ all have the same memory in round $r_{\rho_{U^{23\beta}}}$, which means that the symmetry index of the team is at least equal to the number of nodes occupied by these tokens in round $r_{\rho_{U^{23\beta}}}$. However, since the team is gatherable, Theorem~\ref{theo:charac} implies that the symmetry index of the team is $1$. Consequently all the agents occupies the same node in round $r_{\rho_{U^{23\beta}}}$ and, in view of Lemma~\ref{lem:notwotok}, only agent ${\tt tok}(A)$ is in state {\tt token} in this round. By Lemma~\ref{lem:transit}, the description of state {\tt token} and the fact that no agent declares termination before round $r$, we know that, in round $r_{\rho_{U^{23\beta}}}$, an agent $X$ cannot be in state {\tt explorer} if ${\tt tok}(X)$ is not in state {\tt token} in this round. This means that every agent different from $A$ is not in state {\tt explorer} in round  $r_{\rho_{U^{23\beta}}}$.

From the above explanations, it follows that, in round $r_{\rho_{U^{23\beta}}}$, every agent different from $A$ and ${\tt tok}(A)$ is in state {\tt shadow} and has ${\tt tok}(A)$ as its guide. Since ${\tt tok}(A)$ is in state {\tt token} from round $r_{\rho_{U^{23\beta}}}$ to round $r$ included, we know by the description of state {\tt shadow} and Lemma~\ref{lem:guide} that every agent different from $A$ and ${\tt tok}(A)$ is still in state {\tt shadow} in round $r$. This completes the proof of the lemma.
\end{proof}

With the last two lemmas below, we are ready to wrap up the proof. Specifically, the former states that the problem is solved when a sequence of $U^{25\beta}$ consecutive identical steps is completed by an agent in state {\tt explorer}, while the latter guarantees that such an event occurs after at most a polynomial time in $n$ and $|\lambda|$ from $r_0$.

\begin{lemma}
\label{lem:gatheroccurs}
When a sequence of $U^{25\beta}$ consecutive identical steps is completed in a round $r$ by an agent in state {\tt explorer}, then, in this round, all agents are together and declare termination.
\end{lemma}

\begin{proof}
Without loss of generality, we assume that $r$ is the first round when a sequence of $U^{25\beta}$ consecutive identical steps is completed by an agent $A$ in state {\tt explorer}. We know by lines~\ref{l:repeat}-\ref{l:afterrepeat2} of Algorithm~\ref{alg:algexplo}, that in each of these steps, function ${\tt EST}^+$ returns a triple whose the first element is the Boolean value false. Given that a step is made of an execution of function ${\tt EST}^+$ followed by a possible waiting period, we also know by the third property of Lemma~\ref{lem:exp} that agent $A$ is at ${\tt home}(A)$ in round $r$. Thus, by lines~\ref{l:repeat}-\ref{l:afterrepeat3} of Algorithm~\ref{alg:algexplo}, agent $A$ declares termination at ${\tt home}(A)$ in round $r$.

In view of Lemma~\ref{lem:important} and the fact that while in state {\tt token} an agent remains idle, ${\tt tok}(A)$ is in state {\tt token} at ${\tt home}(A)$ in round $r$. Furthermore, for the agents, other than $A$ and ${\tt tok}(A)$, we know by Lemma~\ref{lem:important} that they are in state {\tt shadow} in round $r$, which means that agent $A$ is the only agent in state {\tt explorer} in round $r$. Thus, by the description of state {\tt token} and the fact that $A$ declares termination at ${\tt home}(A)$ in round $r$, ${\tt tok}(A)$ also declares termination in this round.

We mentioned just above that all the agents, except $A$ and ${\tt tok}(A)$, are in state {\tt shadow} in round $r$. By Lemma~\ref{lem:guide} and the fact that agent ${\tt tok}(A)$ is the only agent in state {\tt token} in round $r$, we also know that these agents occupy in round $r$ the same node as their guide i.e., ${\tt tok}(A)$. In view of the description of state {\tt shadow}, it follows that, in round $r$, these agents in state {\tt shadow} are all at ${\tt home}(A)$ and declare termination, as their guide does so in the same round.

Consequently, we proved that all agents, including $A$ and ${\tt tok}(A)$, are at ${\tt home}(A)$ in round $r$ and declare termination in this round, which ends the proof of the lemma.
\end{proof}

\begin{lemma}
\label{lem:gathertime}
A sequence of $U^{25\beta}$ consecutive identical steps is completed by an agent in state {\tt explorer} within at most a polynomial number of rounds in $n$ and $|\lambda|$ after round $r_0$. 
\end{lemma}

\begin{proof}
Recall that $r_1$ denotes the first round in which at least one agent enters state {\tt explorer}. By Lemmas~\ref{lem:tokappears} and~\ref{lem:notwotok}, the difference $r_1-r_0$ is at most some polynomial $P(n,|\lambda|)$. 

Let $EXP$ be the non-empty set of agents that enters state {\tt explorer} in round $r_1$. By Lemma~\ref{claim2}, we know there exists an agent $A^*$ of $EXP$ that never switches from state {\tt explorer} to state {\tt searcher}.

In view of lines~\ref{l:4}-\ref{l:afterrepeat2} of Algorithm~\ref{alg:algexplo}, this implies that each execution of ${\tt EST}^+$ by agent $A^*$ returns a triple whose the first element is the Boolean value false, and is immediately followed by an execution of line~\ref{l:if2} of Algorithm~\ref{alg:algexplo}. Furthermore, given two triples returned by any two executions of ${\tt EST}^+$ by agent $A^*$, if their third elements (which correspond to the traces of the execution of ${\tt EST}^+$) are identical, then their second elements (which correspond to the orders of the BFS trees constructed during the executions of ${\tt EST}^+$) are also identical and, in both cases, the waiting period at line~\ref{l:if2} of Algorithm~\ref{alg:algexplo} that immediately follows the execution of ${\tt EST}^+$ lasts the same amount of time. From this, it follows that a sequence of $U^{25\beta}$ consecutive steps executed by agent $A^*$ are identical if, and only if, the third elements of the triples returned by the executions of ${\tt EST}^+$ within these steps are all identical. In view of lines~\ref{l:if}-\ref{l:repeat} of Algorithm~\ref{alg:algexplo}, it also follows that agent $A^*$ stops executing the repeat loop of this algorithm when, and only when, it has completed a sequence of $U^{25\beta}$ consecutive identical steps.

Consider any sequence of $U^{25\beta}$ consecutive steps that is started (resp. completed) by agent $A^*$ in some round $r$ (resp. $r'$). If the $U^{25\beta}$ steps of this sequence are not all identical then, according to the explanations given above, there must exist two triples, among those returned during the execution of the sequence, whose the third elements differ. In view of function ${\tt EST}^+$, this can happen only if there is a round $r\leq r'' \leq r'$ such that at least one of the following conditions is satisfied:
\begin{enumerate}
\item $M_{r''}({\tt tok}(A))=M_{r''}(B)$ and $M_{r''-1}({\tt tok}(A))\ne M_{r''-1}(B)$ for some agent $B$
\item $M_{r''}({\tt tok}(A))\ne M_{r''}(B)$ and $M_{r''-1}({\tt tok}(A))=M_{r''-1}(B)$ for some agent $B$
\end{enumerate}

However, by Lemma~\ref{lem:prec}, we know that the first condition cannot actually be satisfied. By the same lemma and the fact that there are at most $n$ agents, we also know that the number of rounds that can fulfilled the second condition is at most $n$. Moreover, $U=2^{(2^{\lceil\log\log (\mu +1)\rceil})}\leq (n+1)^2$ and, by Lemma~\ref{lem:stepduration}, each execution of a step by agent $A^*$ lasts at most $2^{13}(nUP(n,|\lambda|))^{11\beta}$ rounds (we can indeed apply Lemma~\ref{lem:stepduration} because, as mentioned above, each execution of ${\tt EST}^+$ by $A^*$ returns $(\false,*,*)$). This implies that a sequence of $U^{25\beta}$ consecutive steps executed by agent $A^*$ lasts a number of rounds that is at most some polynomial $Q(n,|\lambda|)$ and agent $A^*$ can execute at most $n$ pairwise disjoint such sequences, in which all the $U^{25\beta}$ steps are not identical. Hence, since agent $A^*$ stops executing the repeat loop of Algorithm~\ref{alg:algexplo} when, and only when, it has completed a sequence of $U^{25\beta}$ consecutive identical steps, it must have completed such a sequence by round $r_1+(n+1)Q(n,|\lambda|)\leq P(n,|\lambda|)+(n+1)Q(n,|\lambda|)$, which proves the lemma.


\end{proof}

Lemmas~\ref{lem:gatheroccurs} and~\ref{lem:gathertime} directly imply the next theorem.

\begin{theorem}
\label{theo:pos}
Let $\mathcal{O}$ be the oracle that, for each team $L$ with multiplicity index $\mu$, associates the binary representation of $\lceil\log\log (\mu+1)\rceil$ to each initial setting $IS(L)$ (the size of $\mathcal{K}$ is therefore in ${O}(\log \log \log \mu)$). Algorithm ${\mathcal{HG}}$ is an $\mathcal{O}$-universal gathering algorithm and its time complexity with $\mathcal{O}$ is polynomial in $n$ and $|\lambda|$.
\end{theorem}

Note that any team with multiplicity index $1$ also has symmetry index $1$, and is thus gatherable by Theorem~\ref{theo:charac}.
As a result, from Algorithm ${\mathcal{HG}}$, we can derive another algorithm, call it ${\mathcal{HG}}^+$, that allows to gather any team of multiplicity index $1$ without requiring any initial common knowledge. Specifically, this can be achieved by making the following two modifications to Algorithm~\ref{alg:algexplo}:
\begin{enumerate}
\item Assign value $2$ instead of $2^{(2^x)}$ to variable $U$ at line~\ref{l:3} ($2^{(2^{\lceil\log\log (\mu+1)\rceil})}=2$ when $\mu=1$).
\item Remove line~\ref{l:2}.
\end{enumerate}

By definition, a team for which the multiplicity index is $1$ is a team in which all agents' labels are pairwise distinct. Hence, in view of Theorem~\ref{theo:pos}, we obtain the corollary below.

\begin{corollary}
\label{cor:pos}
Without using any common knowledge $\mathcal{K}$, Algorithm ${\mathcal{HG}}^+$ allows to gather, in polynomial time in $n$ and $|\lambda|$, any team in which all agents' labels are pairwise distinct.
\end{corollary}

\section{Time Complexity vs Size of $\mathcal{K}$: A Negative Result}
\label{sec:neg}

In this section, we establish a negative result on the size of
$\mathcal{K}=\mathcal{O}(IS(L))$ required for an algorithm to gather
any gatherable team $L$ in a time at most polynomial in $n$ and
$|\lambda|$, regardless of the adversary's choices. Essentially, it
states that to gather any gatherable team in polynomial time in $n$
and $|\lambda|$, the size of advice used by Algorithm ${\mathcal{HG}}$
presented in the previous section is almost optimal. Its proof follows
an approach similar to~\cite{CGM12j}.

\begin{theorem}
\label{theo:neg}
%
Let $\mathcal{O}$ be an oracle such that the size of
$\mathcal{K}=\mathcal{O}(IS(L))$, for any gatherable team $L$ and any
initial setting $IS(L)$, is in $o(\log\log\log \mu)$, where $\mu$ is
the multiplicity index of $L$. There exists no $\mathcal{O}$-universal
gathering algorithm whose time complexity with $\mathcal{O}$ is
polynomial in $n$ and $|\lambda|$.
\end{theorem}

\begin{proof}
  Let $\mathcal{O}$ be an oracle that gives an advice
  of size $f(\mu)$ in $o(\log\log\log \mu)$
  for any initial setting $IS(L)$.
  Since $f(\mu) = o(\log\log\log \mu)$, for every $\varepsilon>0$,
  there exists a constant $\mu_0$ such that for every $\mu \geq \mu_0
  ,f(\mu)\leq \varepsilon. \log\log\log \mu$. Hence, we may assume
  that for any initial setting $IS(L)$ with a sufficientally large
  $\mu$, we can fix a constant $c>0$ to be defined later such that the
  size of the advice is less than $\left\lceil c.\log\log\log
  (\mu)\right\rceil$.  Let assume for the sake of contradiction that
  there exists an $\mathcal{O}$-universal gathering algorithm
  $\mathcal{A}$ whose time complexity with oracle $\mathcal{O}$ is
  less than $b.n^d|\lambda|^d$ for some integer constants $b, d\geq
  1$.

We define a set of initial settings $\mathcal{I}=\{IS(L_i)\mid i\in\mathbb{N}^*\}$.
Each initial setting $IS(L_i)\in \mathcal{I}$ corresponds to some ring $R_i$ of size $n_i\geq 3$, with nodes $x_0, x_1, \ldots, x_{n_i-1}$ arranged in clockwise order with exactly one agent on each node. Each node $x_j$ has two ports: port 0 leads to $x_{j+1}$, and port 1 leads to $x_{j-1}$ (indices are taken modulo $n_i$). Each initial setting $IS(L_i)$ is represented by a word $W_i$ of length $n_i$ on the alphabet of labels such that the $j$-th letter of $W_i$ corresponds to the label of the agent placed in node $x_{j-1}$ of $R_i$. We define words $W_i$ recursively. We set $W_1=112$ and so $IS(L_1)$ corresponds to the ring of size $3$ with the first two nodes each containing an agent with label $1$ and the third node containing an agent with label $2$. Let $P(n)=4.b.n^{d-1}$. For $i\in\mathbb{N}^*$, we set $W_{i+1}=(W_i)^{P(|W_i|)}.\ell$ with $\ell$ the label $i+2$. That is, $W_{i+1}$ is constructed by taking $P(|W_i|)$ copies of $W_i$ and adding the label $i+2$. First, we show some properties on each $L_i$ and $W_i$.

\begin{claim}
  \label{claim_L_i}
The following properties are true for all $i\in \mathbb{N}^*$ :

\begin{enumerate}
  \item $L_i$ is gatherable,
  \item $|W_{i+1}| = 4b|W_i|^{d}+1$ and
  \item the length of the binary representation of the smallest label in $L_i$ is 1.
\end{enumerate}
\end{claim}

\begin{proofclaim}
\begin{enumerate}
\item For all $i\in\mathbb{N}^*$, $L_i$ has a symmetry index equal to $1$ since it contains exactly one instance of label $i+1$. It follows by Theorem~\ref{theo:charac} that $L_i$ is gatherable.
\item Observe that $|W_{i+1}|=P(|W_i|)|W_{i}|+1= 4b.|W_i|^{d-1}|W_i|+1= 4b.|W_i|^{d}+1$.
\item For all $i\in \mathbb{N}^*$, the smallest label of $L_i$ is 1 and
  so is the length of its binary representation.
\end{enumerate}
\end{proofclaim}

Let $i\in \mathbb{N}^*$.
By Claim~\ref{claim_L_i}, we have $|W_{i+1}|=4b|W_i|^{d}+1$. Hence, we have $|W_{i+1}|\leq |W_i|^{d+b+3}$.
Now, since $|W_1|=3$, we have $|W_{i}|\leq 3^{(d+b+3)^i}$ and so $\mu_i\leq 3^{(d+b+3)^i}$, where $\mu_i$ is the multiplicity index of $L_i$.
Hence, it follows that there exists a positive constant $c'$ such that we have $\log\log\mu_i \leq c'.i$ for all $i\geq 1$.
We fix the constant $c$ such that the number of binary words of length at most $\left\lceil c\log \left(m\right) \right\rceil $ is less than $\left\lfloor \frac{m}{c'}\right\rfloor$ for a sufficiently large positive integer $m$.
Let $\mathcal{I}_m=\left\{IS(L_i)\mid 1\leq i\leq \left\lfloor \frac{m}{c'}\right\rfloor\right\}$. Observe that for all $IS(L_i)\in\mathcal{I}_m$, $\log\log\mu_i \leq c'.i\leq c'\left\lfloor \frac{m}{c'}\right\rfloor\leq m$.
Hence, for $IS(L_i)\in \mathcal{I}_m$, the size of the advice $\mathcal{K}$ given by $\mathcal{O}$ is at most $\left\lceil c.\log\log\log (\mu_i)\right\rceil\leq \left\lceil c.\log m\right\rceil$.
By the choice of $c$, the number of different advices given by $\mathcal{O}$ to all $IS(L_i)\in\mathcal{I}_m$ is less than $\left\lfloor \frac{m}{c'}\right\rfloor$. Since $|\mathcal{I}_m| = \left\lfloor \frac{m}{c'}\right\rfloor$, it follows that there are two initial settings $IS(L_s)$ and $IS(L_q)$ in $\mathcal{I}_m$ that both receive the same advice. The remainder of the proof then consists in showing that the fact that the two initial settings $IS(L_s)$ and $IS(L_q)$ receive the same advice makes Algorithm $\mathcal{A}$ fails.
To that goal, we can assume w.l.o.g. that $s<q$. Observe that, by construction, $W_q$ contains occurrences of word $(W_s)^{P(|W_s|)}$; let $W$ be the first of these occurrences.  
Let $z_{i}$ be the node of $IS(L_s)$ corresponding to the $(i+1)$-th letter of the word $W_s$ and let $y_{j}$ be the node of $IS(L_q)$ corresponding to the $(j+1)$-th letter of the word $W$. Let $Z_i^r$ (resp. $Y_j^r$) denote the set of agents located at node $z_i$ (resp. $y_j$) in round $r$. In order to show a contradiction, we show the following claim.

\begin{claim}
  \label{claim_4_1}
The following property is true for all round $r$ such that $1\leq r\leq |W|/4$.

$H(r)$: For all $0 \leq i \leq |W_s|-1$ and all $r-1\leq j\leq |W|-r$ such that $i \equiv j \mod |W_s|$, there exists a bijection $f_{i,j}^r : Z_i^r \to Y_j^r$ such that for every agent $Z \in Z_i^r$, the memory state of $Z$ in round $r$ equals that of $f_{i,j}^r(Z)$.
\end{claim}

\begin{proofclaim}
We will show $H(r)$ by induction on $r$. First, we show the base
case $H(1)$ of the induction.  Let $i\in \{0, \dots, |W_s|-1\}$ and
$j\in \{0, \dots, |W|-1\}$ such that $i \equiv j \mod |W_s|$.  The
nodes $z_{i}$ and $y_{j}$ both starts with exactly one agent having
the $(i+1)$-th letter of $W_s$ as  label.  Since at round $1$ the
memory of an agent only contains its label, the singleton sets $Z_i^1$ and $Y_j^1$ contain agents with
identical memory. The bijection $f_{i,j}^1$ is trivial in this case,
so $H(1)$ holds.

Assume $H(r)$ holds for some round $r$ such that $1\leq r\leq
|W|/4-1$. We show that $H(r+1)$ also holds.
By the inductive hypothesis $H(r)$, for all $i\in \{0, \dots,
|W_s|-1\}$ and all $j\in \{r-1, \dots, |W|-r\}$ such that $i \equiv j
\mod |W_s|$, there exists a bijection $f_{i,j}^r : Z_i^r \to Y_j^r$
such that for every agent $Z \in Z_i^r$, the memory state of $Z$ in
round $r$ equals that of $f_{i,j}^r(Z)$.
Now, $\{r-1, \dots, |W|-r\} = \{r-1\} \cup \{(r+1)-1, \dots, |W|-(r+1)\} \cup \{|W|-r\}$. So, for all
%
$k \in \{(r+1)-1, \dots, |W|-(r+1)\}$,
%
$\forall Y \in Y_k^{r+1}$, $Y \in \cup_{j \in \{r-1, \dots, |W|-r\}} Y_j^{r}$.
Since the algorithm is deterministic, agents with the same advice and the same memory take the same action.
Thus, from the bijections $f_{i,j}^r$ (with $i\in \{0, \dots,
|W_s|-1\}$, $j\in \{r-1, \dots, |W|-r\}$, and $i \equiv j \mod
|W_s|$), we can deduce that for all $i\in \{0, \dots, |W_s|-1\}$ and
all $k \in \{(r+1)-1, \dots, |W|-(r+1)\}$ such that $i \equiv k \mod
|W_s|$, there exists a bijection $f_{i,k}^{r+1} : Z_i^{r+1} \to
Y_j^{r+1}$ such that for every agent $Z \in Z_i^{r+1}$, the memory
state of $Z$ in round $r+1$ equals that of $f_{i,k}^{r+1}(Z)$, and we
are done.
\end{proofclaim}

Let $\lambda_s$ be the smallest label of $L_s$.  There exists $i \in
\{0, \dots, |W_s|-1\}$ such that agents in $IS(L_s)$ must declare
termination at node $z_i$ at some round $r\leq|W|/4$.  Indeed,
algorithm $\mathcal{A}$ applied on $IS(L_s)$ has time complexity less
than $b.(n_s)^d.(|\lambda_s|)^d=b.(n_s)^d=n_s.P(n_s)/4= |W|/4$ as
$|\lambda_s|=1$ by Claim~\ref{claim_L_i}. Since $|W|-2r+1 \geq
\frac{|W|}{2}\geq \frac{P(|W_s|)|W_s|}{2}\geq 2|W_s|$, there are two
distinct integers $j$ and $j'$ such that $r-1\leq j,j'\leq |W|-r$, $i
\equiv j \mod |W_s|$ and $i \equiv j' \mod |W_s|$.  By
Claim~\ref{claim_4_1}, there exist two bijections $f_{i,j}^r$ and
$f_{i,j'}^r$ such that for every $Z \in Z_i^r$, the memory state of
$Z$ at round $r$ equals that of $f_{i,j}^r(Z)$ and $f_{i,j'}^r(Z)$. It
follows that $f_{i,j}^r(Z)$ and $f_{i,j'}^r(Z)$ must also declare
termination at round $r$ when algorithm $\mathcal{A}$ applied on
$IS(L_q)$. Since $f_{i,j}^r(Z)$ and $f_{i,j'}^r(Z)$ are on distinct
nodes $y_j$ and $y_{j'}$, we obtain a contradiction with the fact that
$\mathcal{A}$ is an $\mathcal{O}$-universal gathering algorithm.
\end{proof}

\section{Conclusion}

Until now, the problem of gathering had been studied at two extremes: either when agents are equipped with pairwise distinct labels, or when they are entirely anonymous. Our paper bridged this gap by considering the more general context of labeled agents with possible homonyms. In this context, we fully characterized the teams that are gatherable and we analyzed the question of whether a single algorithm can gather all of them in time polynomial  in $n$ and $|\lambda|$. This question was natural given the well-known poly$(n,|\lambda|)$-time lower bound of \cite{DessmarkFKP06} for gathering just two agents with distinct labels without initial common knowledge. On the negative side, we showed that no algorithm can gather every gatherable team in poly$(n,|\lambda|)$ time, even if they initially share $o(\log \log \log \mu)$ bits of common knowledge, where $\mu$ is the multiplicity index of the team. On the positive side, we designed an algorithm that gathers all of them in poly$(n,|\lambda|)$ time using only $O(\log \log \log \mu)$ bits of common knowledge.

As a general open question, it is legitimate to ask what happens to complexity, or even simply feasibility, when we vary the amount of shared knowledge across the full spectrum that ranges from no knowledge to complete knowledge of the initial setting. Actually, using the same arguments as in the proof of Theorem~\ref{theo:neg}, we can show that no algorithm can gather all gatherable teams if the agents initially share no knowledge about their initial setting.\footnote{For any algorithm $\mathcal{A}$ supposedly achieving this, one can construct two initial settings $IS(L)$ and $IS(L')$, around two gatherable teams $L$ and $L'$, that would necessarily require two initial different pieces of advice from an oracle to be solved by $\mathcal{A}$. Precisely, this can be done using a design principle similar to the one employed to construct the set of initial settings $\mathcal{I}$ in the proof of Theorem~\ref{theo:neg}. This yields a direct contradiction with the fact that $\mathcal{A}$ can work without any common knowledge.} But what about the rest of the spectrum? Are there sharp threshold effects or only smooth trade-offs? This is an interesting but challenging avenue for future work.

In the classical scenario where all agent labels are pairwise distinct, we obtained the first poly$(n,|\lambda|)$-time algorithm that requires no common knowledge to gather teams of arbitrary size. This result closed a fundamental open problem when agents traverse edges in synchronous rounds. Intriguingly, a similar problem remains open when the assumption of synchrony is removed, and agents instead traverse edges at speeds that an adversary may arbitrarily vary over time (as in \cite{CzyzowiczPL12,DieudonnePV15,MarcoGKKPV06}). It is then no longer a matter of determining complexity in terms of time (since time is entirely controlled by an adversary), but rather in terms of cost, that is, the total number of edge traversals performed by all agents in order to gather. Specifically, with no initial common knowledge, the cost of asynchronous gathering is lower-bounded by a polynomial in $n$ and $|\lambda|$ (cf. \cite{MarcoGKKPV06}), but no algorithm guaranteeing such a polynomial cost complexity under the same knowledge constraint is known, except in the special case of two agents \cite{DieudonnePV15}. As in the synchronous model, the challenge of generalization is closely tied to termination detection. However, the tools we used to tackle this rely on synchronicity, especially waiting periods, which lose their meaning when agents move asynchronously. Thus, developing new techniques to address this problem in this hasher environment constitutes another interesting direction for future research.

\newpage
\bibliographystyle{plain}
\bibliography{biblio}

\begin{thebibliography}{10}

\bibitem{AgmonP06}
Noa Agmon and David Peleg.
\newblock Fault-tolerant gathering algorithms for autonomous mobile robots.
\newblock {\em {SIAM} J. Comput.}, 36(1):56--82, 2006.

\bibitem{Alpern02}
Steve Alpern.
\newblock Rendezvous search: {A} personal perspective.
\newblock {\em Operations Research}, 50(5):772--795, 2002.

\bibitem{Alpern03}
Steve Alpern.
\newblock {\em The theory of search games and rendezvous}.
\newblock International Series in Operations Research and Management Science,
  Kluwer Academic Publishers, 2003.

\bibitem{BouchardDD16}
S{\'{e}}bastien Bouchard, Yoann Dieudonn{\'{e}}, and Bertrand Ducourthial.
\newblock Byzantine gathering in networks.
\newblock {\em Distributed Comput.}, 29(6):435--457, 2016.

\bibitem{BouchardDL22}
S{\'{e}}bastien Bouchard, Yoann Dieudonn{\'{e}}, and Anissa Lamani.
\newblock Byzantine gathering in polynomial time.
\newblock {\em Distributed Comput.}, 35(3):235--263, 2022.

\bibitem{BouchardDP23}
S{\'{e}}bastien Bouchard, Yoann Dieudonn{\'{e}}, and Andrzej Pelc.
\newblock Want to gather? no need to chatter!
\newblock {\em {SIAM} J. Comput.}, 52(2):358--411, 2023.

\bibitem{BouchardDPP18}
S{\'{e}}bastien Bouchard, Yoann Dieudonn{\'{e}}, Andrzej Pelc, and Franck
  Petit.
\newblock On deterministic rendezvous at a node of agents with arbitrary
  velocities.
\newblock {\em Inf. Process. Lett.}, 133:39--43, 2018.

\bibitem{ChalopinDK10}
J{\'{e}}r{\'{e}}mie Chalopin, Shantanu Das, and Adrian Kosowski.
\newblock Constructing a map of an anonymous graph: Applications of universal
  sequences.
\newblock In Chenyang Lu, Toshimitsu Masuzawa, and Mohamed Mosbah, editors,
  {\em Principles of Distributed Systems - 14th International Conference,
  {OPODIS} 2010, Tozeur, Tunisia, December 14-17, 2010. Proceedings}, volume
  6490 of {\em Lecture Notes in Computer Science}, pages 119--134. Springer,
  2010.

\bibitem{CGM12j}
J{\'{e}}r{\'{e}}mie Chalopin, Emmanuel Godard, and Yves M{\'{e}}tivier.
\newblock Election in partially anonymous networks with arbitrary knowledge in
  message passing systems.
\newblock {\em Distributed Comput.}, 25(4):297--311, 2012.

\bibitem{CieliebakFPS12}
Mark Cieliebak, Paola Flocchini, Giuseppe Prencipe, and Nicola Santoro.
\newblock Distributed computing by mobile robots: Gathering.
\newblock {\em {SIAM} J. Comput.}, 41(4):829--879, 2012.

\bibitem{CzyzowiczPL12}
Jurek Czyzowicz, Andrzej Pelc, and Arnaud Labourel.
\newblock How to meet asynchronously (almost) everywhere.
\newblock {\em {ACM} Trans. Algorithms}, 8(4):37:1--37:14, 2012.

\bibitem{DP16j}
Dariusz Dereniowski and Andrzej Pelc.
\newblock Topology recognition and leader election in colored networks.
\newblock {\em Theoretical Computer Science}, 621:92--102, 2016.

\bibitem{DessmarkFKP06}
Anders Dessmark, Pierre Fraigniaud, Dariusz~R. Kowalski, and Andrzej Pelc.
\newblock Deterministic rendezvous in graphs.
\newblock {\em Algorithmica}, 46(1):69--96, 2006.

\bibitem{DieudonneP16}
Yoann Dieudonn{\'{e}} and Andrzej Pelc.
\newblock Anonymous meeting in networks.
\newblock {\em Algorithmica}, 74(2):908--946, 2016.

\bibitem{DieudonnePP12}
Yoann Dieudonn{\'{e}}, Andrzej Pelc, and David Peleg.
\newblock Gathering despite mischief.
\newblock In Yuval Rabani, editor, {\em Proceedings of the Twenty-Third Annual
  {ACM-SIAM} Symposium on Discrete Algorithms, {SODA} 2012, Kyoto, Japan,
  January 17-19, 2012}, pages 527--540. {SIAM}, 2012.

\bibitem{DieudonnePV15}
Yoann Dieudonn{\'{e}}, Andrzej Pelc, and Vincent Villain.
\newblock How to meet asynchronously at polynomial cost.
\newblock {\em {SIAM} J. Comput.}, 44(3):844--867, 2015.

\bibitem{FKKLS04j}
Paola Flocchini, Evangelos Kranakis, Danny Krizanc, Flaminia~L. Luccio, and
  Nicola Santoro.
\newblock Sorting and election in anonymous asynchronous rings.
\newblock {\em Journal of Parallel and Distributed Computing}, 64(2):254--265,
  2004.

\bibitem{FraigniaudIP06}
Pierre Fraigniaud, David Ilcinkas, and Andrzej Pelc.
\newblock Oracle size: a new measure of difficulty for communication tasks.
\newblock In Eric Ruppert and Dahlia Malkhi, editors, {\em Proceedings of the
  Twenty-Fifth Annual {ACM} Symposium on Principles of Distributed Computing,
  {PODC} 2006, Denver, CO, USA, July 23-26, 2006}, pages 179--187. {ACM}, 2006.

\bibitem{GlacetMP17}
Christian Glacet, Avery Miller, and Andrzej Pelc.
\newblock Time vs. information tradeoffs for leader election in anonymous
  trees.
\newblock {\em {ACM} Trans. Algorithms}, 13(3):31:1--31:41, 2017.

\bibitem{GorainMP23}
Barun Gorain, Avery Miller, and Andrzej Pelc.
\newblock Four shades of deterministic leader election in anonymous networks.
\newblock {\em Distributed Comput.}, 36(4):419--449, 2023.

\bibitem{HiroseNOI22}
Jion Hirose, Junya Nakamura, Fukuhito Ooshita, and Michiko Inoue.
\newblock Weakly byzantine gathering with a strong team.
\newblock {\em {IEICE} Trans. Inf. Syst.}, 105-D(3):541--555, 2022.

\bibitem{HiroseNOI24}
Jion Hirose, Junya Nakamura, Fukuhito Ooshita, and Michiko Inoue.
\newblock Fast gathering despite a linear number of weakly byzantine
  agents\({}^{\mbox{{\textdagger}}}\).
\newblock {\em Concurr. Comput. Pract. Exp.}, 36(14), 2024.

\bibitem{IzumiSKIDWY12}
Taisuke Izumi, Samia Souissi, Yoshiaki Katayama, Nobuhiro Inuzuka, Xavier
  D{\'{e}}fago, Koichi Wada, and Masafumi Yamashita.
\newblock The gathering problem for two oblivious robots with unreliable
  compasses.
\newblock {\em {SIAM} J. Comput.}, 41(1):26--46, 2012.

\bibitem{KowalskiM08}
Dariusz~R. Kowalski and Adam Malinowski.
\newblock How to meet in anonymous network.
\newblock {\em Theor. Comput. Sci.}, 399(1-2):141--156, 2008.

\bibitem{KranakisKR06}
Evangelos Kranakis, Danny Krizanc, and Sergio Rajsbaum.
\newblock Mobile agent rendezvous: {A} survey.
\newblock In {\em Structural Information and Communication Complexity, 13th
  International Colloquium, {SIROCCO} 2006, Chester, UK, July 2-5, 2006,
  Proceedings}, pages 1--9, 2006.

\bibitem{MarcoGKKPV06}
Gianluca~De Marco, Luisa Gargano, Evangelos Kranakis, Danny Krizanc, Andrzej
  Pelc, and Ugo Vaccaro.
\newblock Asynchronous deterministic rendezvous in graphs.
\newblock {\em Theor. Comput. Sci.}, 355(3):315--326, 2006.

\bibitem{MillerP15}
Avery Miller and Andrzej Pelc.
\newblock Fast rendezvous with advice.
\newblock {\em Theor. Comput. Sci.}, 608:190--198, 2015.

\bibitem{MollaMM23}
Anisur~Rahaman Molla, Kaushik Mondal, and William K.~Moses Jr.
\newblock Fast deterministic gathering with detection on arbitrary graphs: The
  power of many robots.
\newblock In {\em {IEEE} International Parallel and Distributed Processing
  Symposium, {IPDPS} 2023, St. Petersburg, FL, USA, May 15-19, 2023}, pages
  47--57. {IEEE}, 2023.

\bibitem{NisseS09}
Nicolas Nisse and David Soguet.
\newblock Graph searching with advice.
\newblock {\em Theor. Comput. Sci.}, 410(14):1307--1318, 2009.

\bibitem{Pelc18}
Andrzej Pelc.
\newblock Deterministic gathering with crash faults.
\newblock {\em Networks}, 72(2):182--199, 2018.

\bibitem{Pelc19}
Andrzej Pelc.
\newblock Deterministic rendezvous algorithms.
\newblock In Paola Flocchini, Giuseppe Prencipe, and Nicola Santoro, editors,
  {\em Distributed Computing by Mobile Entities, Current Research in Moving and
  Computing}, volume 11340 of {\em Lecture Notes in Computer Science}, pages
  423--454. Springer, 2019.

\bibitem{PelcY19}
Andrzej Pelc and Ram~Narayan Yadav.
\newblock Using time to break symmetry: Universal deterministic anonymous
  rendezvous.
\newblock In Christian Scheideler and Petra Berenbrink, editors, {\em The 31st
  {ACM} on Symposium on Parallelism in Algorithms and Architectures, {SPAA}
  2019, Phoenix, AZ, USA, June 22-24, 2019}, pages 85--92. {ACM}, 2019.

\bibitem{Reingold08}
Omer Reingold.
\newblock Undirected connectivity in log-space.
\newblock {\em J. {ACM}}, 55(4):17:1--17:24, 2008.

\bibitem{SaxenaM24}
Ashish Saxena and Kaushik Mondal.
\newblock A further study on weak byzantine gathering of mobile agents.
\newblock {\em Theor. Comput. Sci.}, 1022:114892, 2024.

\bibitem{Schelling}
Thomas Schelling.
\newblock {\em The Strategy of Conflict}.
\newblock Oxford University Press, Oxford, 1960.

\bibitem{Ta-ShmaZ14}
Amnon Ta{-}Shma and Uri Zwick.
\newblock Deterministic rendezvous, treasure hunts, and strongly universal
  exploration sequences.
\newblock {\em {ACM} Trans. Algorithms}, 10(3):12:1--12:15, 2014.

\bibitem{YK89c}
Masafumi Yamashita and Tiko Kameda.
\newblock Electing a leader when processor identity numbers are not distinct
  (extended abstract).
\newblock In {\em Proceedings of the 3rd International Workshop on Distributed
  Algorithms}, page 303–314, Berlin, Heidelberg, 1989. Springer-Verlag.

\end{thebibliography}

\end{document}